\documentclass[12pt,a4paper,twoside]{article}

\usepackage{fancyhdr,graphicx,young}
\usepackage{amsmath,amsfonts,amsthm}
\usepackage{psfrag}

\numberwithin{equation}{section}

\DeclareMathOperator{\tr}{Tr}
\DeclareMathOperator{\rpart}{Re}
\DeclareMathOperator{\ipart}{Im}
\DeclareMathOperator{\erfc}{erfc}
\DeclareMathOperator{\diag}{diag}

\newcommand\ie{\emph{i.e. }}
\newcommand\eg{\emph{e.g. }}
\newcommand\ea{\emph{et al.}}
\newcommand{\abs}[1]{\left \lvert #1 \right \rvert}
\newcommand\UN{\mathrm{U}(N)}
\newcommand\mun{\mu^{\mathrm{G}}}
\newcommand\ed{\mathrm{d}}

\theoremstyle{plain}
\newtheorem{theorem}{Theorem}[section]
\newtheorem{proposition}[theorem]{Proposition}
\newtheorem{lemma}[theorem]{Lemma}
\newtheorem{corollary}[theorem]{Corollary}
\newtheorem{remark}[theorem]{Remark}

\newtheorem*{propositiona}{Proposition}

\theoremstyle{definition}

\begin{document}
\title{Rate of convergence of linear functions on the unitary group} 
 \author{J. P. Keating,
  F. Mezzadri and B. Singphu}

\date{}

\maketitle

\begin{abstract}
  We study the rate of convergence to a normal random variable of the
  real and imaginary parts of $\tr A_N U$, where $U$ is an $N \times
  N$ random unitary matrix and $A_N$ is a deterministic complex
  matrix.  We show that the rate of convergence is $O\left(N^{-2 +
      b}\right)$, with $0\le b < 1$, depending only on the asymptotic
  behaviour of the singular values of $A_N$; for example, if the
  singular values are non-degenerate, different from zero and $O(1)$
  as $N \to \infty$, then $b=0$.  The proof uses a Berry-Ess\'een
  inequality for linear combinations of eigenvalues of random unitary
  matrices, and so appropriate for strongly dependent random
  variables.
\end{abstract}

\vspace{.25cm}

\hspace{.26cm} 2010 MSC: 15B52, 60F05

\section{Introduction}
\label{sec:intro}
The value distributions of traces of random unitary matrices have been
studied extensively over the past fifteen
years~\cite{DS94,HKSSZ96,Joh97,Sos00,DE01,HR03,PV04,Sto05}. The main
reason is that they are connected with the linear statistics
\begin{equation}
  \label{eq:lin_stat}
  S_N(\chi) := \chi\left(e^{i\theta_1}\right) + \cdots
  + \chi\left(e^{i\theta_N}\right),
\end{equation}
where $\chi$ is a suitable test function and $e^{i\theta_1},\ldots,
e^{i\theta_N}$ are the eigenvalues of $N \times N$ unitary matrices
$U$ distributed according to Haar measure.  It turns out that in many
applications in particle physics, open quantum systems, quantum
chromodynamics and scattering theory it is interesting to understand
the asymptotic ($N \to \infty$) behaviour not only of $\tr U$ but of
the more general random variable
\begin{equation}
  \label{Eq:def_zn}
Z_N := V_N + iW_N = \tr A_N U,
\end{equation}
where $V_N$ (respectively $W_N$) is the real (imaginary) part of
$Z_N$, and $A_N$ is a deterministic complex matrix. (See,
\textit{e.g.},~\cite{Sam80,Bar80,BB96,SW03} and references therein.)
In other words, we want to understand the distribution of linear
combinations of the elements of random unitary matrices. In general,
this type of question arises when Random Matrix Theory is applied to
non-Hermitian quantum mechanics, an area of physics which has grown
rapidly in the last decades (see, \textit{e.g.}, \cite{Moi98,OPR03}
and references therein).  As we shall see, the invariance of Haar
measure on $\UN$ under group action implies that the distributions of
$V_N$ and $W_N$ are the same.  Therefore, we shall restrict our
attention to $V_N$.

Samuel~\cite{Sam80} and Bars~\cite{Bar80} computed the first few terms
in the cumulant expansion of $V_N$, which implicitly show that it
converges in distribution to a normal random variable when $N \to
\infty$. D'Aristotile \ea~\cite{ADN03} gave a rigorous proof of this
result. Collins and Stolz~\cite{CS08} proved a multivariate version of
this theorem: they showed that a vector of the form
\begin{equation}
  \label{eq:multivariate}
  \left(\tr A^{(1)}_N U,\dotsc,\tr A^{(r)} U\right),
\end{equation}
where $r$ is independent of $N$, converges to a joint normal
distribution.

In her PhD thesis, Meckes~\cite{Me06,Me08} studied the rate of
convergence of $V_N$ to a central limit theorem using Stein's method
of exchangeable pairs.  Let us normalise $A_N$ so that $\tr A_N A_N^*
= N$, where $A_N^*$ is the conjugate transpose of $A_N$, and denote by
$\mathcal{N}(\mu,\sigma^2)$ a normal random variable with mean $\mu$
and variance $\sigma^2$.  Meckes proved that the distance of $V_N$ to
$\mathcal{N}(0,1/2)$ in the total variation metric on probability
measures is bounded by $c_NN^{-1}$, where $c_N$ is asymptotic to
$2\sqrt{2}$. Chatterjee and Meckes~\cite{CM08} obtained a rate of
order $O(N^{-1})$ in the multivariate setting too, and showed that the
constant is linear in $r$.

The bound computed by Meckes holds for any $A_N \in \mathbb{C}^{N
  \times N}$, subject to the constraint $\tr A_NA_N^* = N$.  However,
given a fixed sequence $\left \{ A_N \right \}_{N > 1}$, it is natural
to ask how the rate of convergence of $V_N$ depends on $A_N$. The
purpose of this paper is to show that this rate is $O\left(N^{-2 + b
  }\right)$, where $0\le b < 1$, depending only on the leading order
asymptotics as $N \to \infty$ of the greatest singular value of $A_N$.
For example, if the elements of $A_N$ do not grow with $N$ --- which
is what one would expect for a generic sequence $\left \{A_N\right
\}_{N > 1}$ --- then $b=0$ and the rate of convergence is
$O\left(N^{-2}\right)$.  When $b=1$ only a finite fraction of the
singular values is different from zero in the limit $N \to
\infty$. For technical reasons, which we will discuss in
section~\ref{sse:mom_calc}, we exclude the case $b=1$.  Meckes' bound
$c_N N^{-1}$ does not discern the dependence of the rate of
convergence on the singular values of $A_N$, and our result implies
that it is sharp only when $b=1$.

Our approach is based on the method of moments, which allows us to
prove a Berry-Ess\'een inequality for the eigenvalues of random
unitary matrices.  In general, Berry-Ess\'een bounds are used to prove
central limit theorems for sums of independent or weakly dependent
random variables.  It is notable that such a bound exists for sums
of eigenvalues of matrices in $\UN$, which are strongly correlated.

When $A_N$ is the identity, then $Z_N$ is a class function and the
underlying group structure of $\UN$ can be exploited. For general
$A_N$ these group-theoretical tools are not available. There is a
considerable literature addressing the problem of the distribution of
$\tr U^j$, where $j \in \mathbb{Z}_{+}$. Diaconis and
Shahshahani~\cite{DS94}, and independently Haake \ea~\cite{HKSSZ96},
proved that it convergences in distribution to $\sqrt{j} Z$, where $Z$
is a standard normal complex random variable.  Diaconis and
Shahshahani's proof is based on the method of moments; they showed
that the $k$-th moments of $\tr U^j/\sqrt{j}$ are exactly Gaussian for
$k \le N$.  This property prompted Diaconis to conjecture that the
convergence to a normal random variable is very fast, either
exponential or even superexponential.  Consider the error
\begin{equation}
  \label{eq:error_term}
  e(N) := \sup_{x \in \mathbb{R}}\abs{F_N(x) - \Phi(x)},
\end{equation}
where
\begin{equation}
  \label{eq:Phi}
  \Phi(x) := \frac{1}{\sqrt{2\pi}}\int_{-\infty}^x e^{-\frac{t^2}{2}}dt
\end{equation}
and $F_N(x)$ is the distribution function of $\sqrt{2/j}\rpart \tr
U^j$, \ie
\begin{equation}
  \label{eq:dis_fun}
  F_N(x) :=\int_{-\infty}^x f_N\left(t\right)dt,
\end{equation}
where $f_N$ is the probability density function
(\textit{p.d.f.}). Johannson~\cite{Joh97} proved that $e(N)=
O_\epsilon\left(N^{-\epsilon N}\right)$.  He also showed that the
distance of $\sqrt{2/j}\rpart \tr U^j$ to $\mathcal{N}(0,1)$ in the
total variation norm is of the same order.  Such a rate of convergence
to a central limit theorem is unusual in probability theory.  The
approach that we use to achieve our bounds also sheds light on why the
convergence of $\tr U^j $ is so fast.

Subsequently, many authors have refined or improved Diaconis and
Shahshahani's results. Soshnikov~\cite{Sos00} showed that the linear
statistics~\eqref{eq:lin_stat} converge in distribution to a normal
random variable in the \textit{mesoscopic} regime too, \ie if one
considers eigenvalues in an arc of length $L_N$ with $L_N/N \to 0$ as
$N \to \infty$. Hughes and Rudnick~\cite{HR03} studied the
\textit{scaling limit} $L_N=N$.  It turns out that the number of
moments of $S_N(\chi)$ that are exactly Gaussian depends on the class
of test functions considered.  Diaconis and Evans~\cite{DE01} used the
results in~\cite{DS94} to study the asymptotic distributions of
integrals of the type $\int_{\mathbb{S}}fd\Xi$, where $\mathbb{S}$ is
the unit circle and $\Xi_N$ is the random point measure that places a
unit mass at each eigenvalue $e^{i\theta_j}$. Pastur and
Vasilchuk~\cite{PV04} and Stolz~\cite{Sto05} gave alternative proofs
of the convergence to normal random variables of $\tr U^j$.

This article is structured as follows.  In
\S\ref{se:statement_results} we discuss the background of the problem
and introduce our main results. The moments and cumulants of $V_N$ can
be computed using the character theory of the symmetric group; these
calculations are detailed in~\S\ref{sec:mom_cum}.
In~\S\ref{sec:proof_be} we present the proof of the Berry-Ess\'een
inequality.  Finally, \S\ref{sec:proof_main_th} and
\S\ref{se:tot_proof} are devoted to the proofs of the main theorems.

% We compute a bound for the rate of convergence of $\tr A_N U$ using
% the cumulant expansion computed by Samuel~\cite{Sam80} an
% Bars~\cite{Bar80}.  The knowledge of the cumulants gives a
% Berry-Ess\'een type inequality that leads to optimal bounds.  Indeed,
% we prove that the distance~\eqref{eq:error_term} is of order
% $O\left(N^{-2}\right)$ and that the bound in the total variation norm
% is $O_\epsilon\left(N^{-2 + \epsilon}\right)$.

\section{Statement of results}
\label{se:statement_results}

\subsection{Preliminaries}
Let us introduce the random variables
\begin{equation}
  \label{eq:def_XN_YN}
  X_N := \left(\rpart \tr A_N U\right)/\sigma = V_N/\sigma\quad
\text{and} \quad  Y_N :=\left(\ipart \tr A_N U\right)/\sigma 
= W_N/\sigma,
\end{equation}
where $U$ is an $N \times N$ unitary matrix distributed according to
Haar measure and
\begin{equation}
   \label{eq:renormalization}
   \sigma^2 := \frac{\tr A_NA_N^*}{2N}.
\end{equation}
The matrices in a given sequence $\left \{A_N\right \}_{N >1}$ can be
normalised so that $\sigma^2$ is independent of $N$.

Using the polar decomposition we can factorize $A_N$ in the product
\begin{equation}
  \label{eq:sing_value_dec}
  A_N = H_NV,
\end{equation}
where $V \in \UN$ and $H_N = \sqrt{A_N A_N^*}$ is
positive-semidefinite.  Let us also write $U =W\Theta W^*$, where $W
\in \UN$ and $\Theta =
\diag\left(e^{i\theta_1},\dotsc,e^{i\theta_N}\right)$.  Since Haar
measure is invariant under group action, the random variable $\tr H_N
U/\sigma$ has the same distribution as $\tr A_N U/\sigma$.  Thus,
without loss of generality, we can restrict $A_N$ to the set of
positive-semidefinite matrices. Furthermore, we have
\begin{equation}
  \label{eq:trace_diag}
  \tr A_N U = \tr A_N W\Theta W^* = \tr W^* A_N W
  \Theta = \tr \tilde A_N \Theta = \sigma \sum_{j=1}^Na_j e^{i\theta_j},
\end{equation}
where $\tilde A_N$ is Hermitian positive-semidefinite too and $\sigma
a_j  \ge 0$ are its diagonal elements.  Therefore, we can
write
\begin{subequations}
  \label{eq:trace_like}
  \begin{align}
    \label{eq:real_p}
     X_N & = a_1 \cos \theta_1 + \dotsb + a_N \cos \theta_N, \\
    \label{eq:imag_p}
     Y_N &  = a_1 \sin \theta_1 + \dotsb + a_N \sin \theta_N.
  \end{align}
\end{subequations}
Since Haar measure is invariant under translation, $X_N$ and $Y_N$ have
the same probability distribution. Thus, we shall  restrict our
attention to $X_N$.

The characteristic function of $X_N$ is defined by
\begin{equation}
  \label{eq:def_char_fun}
  \psi_N(\xi) := \mathbb{E}_{\UN}\left\{e^{i\xi X_N}\right\}.
\end{equation}
It admits a representation as an integral over the unitary group.  We
have
\begin{equation}
  \label{eq:iz_type_int}
  \psi_N(\xi) = \int_{\UN}\exp\left(\frac{i\xi}{2\sigma}\Bigl(\tr A_N
    U + \tr A_N^* U^*\Bigr)\right)d\mu_{\mathrm{H}}(U),
\end{equation}
where $d\mu_{\mathrm{H}}$ denotes Haar measure over $\UN$.  When $A_N$
is not singular, such an integral can be evaluated
explicitly~\cite{BRT81} (see also~\cite{SW03} when the matrix in the
second trace is different from $A_N^*$):
\begin{equation}
  \label{eq:psiN_explicit}
  \psi_N(\xi) =
  \left(\frac{2\sigma}{\xi}\right)^{\frac{N\left(N-1\right)}{2}}
  \left(\prod_{j=1}^{N-1}j!\right) \frac{\det_{N\times
      N}\left(\nu_j^{k-1}J_{k-1}\left(\xi \nu_j/
          \sigma\right)\right)}{\prod_{1 \le j < k \le N}\left(\nu_k^2
      - \nu_j^2\right)},
\end{equation}
where $\nu_1,\dotsc,\nu_N$ are the singular values of $A_N$ and $J_k$
is the Bessel function of the first kind.  Unfortunately, this
beautiful formula is not the best starting point for a straightforward
asymptotic analysis.  In order to determine the rate of convergence of
$X_N$, we will need to control $\psi_N(\xi)$ when $\xi$ grows like a
power of $N$.  This means that $N$ appears as a parameter in both the
argument and the index of the Bessel functions. The facts that the
asymptotic limit of $J_k(x)$ as $x \to \infty$ is not uniform in the
index, and that all the Bessel functions from $J_0$ to $J_{N-1}$
appear in the determinant, render the analysis of
formula~\eqref{eq:psiN_explicit} difficult.

Damgaard and Splittorff~\cite{DS00} computed the first few terms of
the asymptotic expansions of integral~\eqref{eq:iz_type_int} for
``low-mass'' and ``large-mass''.  In our formalism, this means in the
limit as $\xi \to 0$ and $\xi \to \infty$.

The approach that we adopt is based on the method of moments, which
can be computed explicitly up to the $2N$-th for any matrix $A_N$,
whether singular or not.  The only constraint that we impose on the
sequence $\left \{A_N \right \}_{N > 1}$ is the
normalisation~\eqref{eq:renormalization}.

Our results will depend on the asymptotic properties of the singular
values of $A_N$.  Therefore, we need to introduce quantities that
characterise their behaviour in the limit as $N \to \infty$. Let us
order the singular values of $A_N$ so that $\nu_1 \le \dotsb \le
\nu_N$ and let $\nu_N^2 = O\left(N^b\right)$.  Then, define
\begin{equation}
  \label{eq:bound}
  k := \inf \left \{ c\in \mathbb{R}  \lvert \: \nu_N^2 \le cN^b, \;
    \forall N > 1\right \}.
\end{equation}
Since $b$ is optimal, the normalization~\eqref{eq:renormalization}
implies that $0\le b \le 1$ and $k > 0$. The meaning of $b$ and $k$
can be illustrated with a few examples.  If all the matrix elements of
$A_N$ are $O(1)$ as $N \to \infty$, then $b=0$. Alternatively,
consider the sequence of matrices
\begin{equation}
  \label{eq:ex_1}
  A_N = \diag \left(\sqrt{2N},0,\dotsc,0\right).
\end{equation}
Then $b=1$ and $k=2$. In other words, $b$ not only gives the rate of
growth of $\nu_N$, but also measures how sparse the set of singular
values is in the limit $N \to \infty$.

% \begin{subequations}
% \label{eq:cum_2}
% \begin{align}
%   \kappa_2^N & = 1, \\
%   \label{eq:4th_cumulant}
%  \kappa_4^N & = - \frac{3x_2}{N^2x_1^{2}\left(1 - 1/N^2\right)},\\
%   \kappa_6^N &= \frac{60 x_3}{N^4 x_1^{3}\left(1 - 1/N^2\right)\left(1 -
%       4/N^2\right)}.
% \end{align}
% \end{subequations}

\subsection{Rates of convergence}

Using the same notation as in \S\ref{sec:intro}, $F_N$ and $\Phi$ will
denote the distribution functions of $X_N$ and of a standard normal
random variable respectively; similarly, $f_N$ is the \textit{p.d.f.}
Furthermore, we shall write
\begin{equation}
  \label{eq:normal_pdf_cf}
  \phi(x) := \frac{e^{-x^2/2}}{\sqrt{2\pi}}\quad  \text{and} \quad 
  \psi(\xi) :=  \int_{-\infty}^\infty e^{i\xi x} \phi(x)dx = e^{-\xi^2/2}. 
\end{equation}

\begin{theorem}
\label{th:main_theorem}
Suppose $\left\{A_N \right \}_{N > 1}$ is a sequence of matrices such
that $\sigma^2= \tr A_N A_N^* / (2N)$ is independent of $N$ and that
$0\le b <1$. We have
\begin{equation}
  \label{eq:error}
  e\left(N\right) := \sup_{x \in \mathbb{R}} \left \lvert F_N(x) -
    \Phi(x) \right \rvert = O\left( N^{-2 + b}\right), \quad N \to
  \infty.
\end{equation}
\end{theorem}
As we shall see, the power of minus two in~\eqref{eq:error} is
determined by the Haar measure on $\UN$. The sequence $\left \{ A_N
\right \}_{N > 1}$ influences the rate of convergence only through the
parameter $b$, which is a measure of the asymptotic distribution of
the singular values of the matrices $A_N$.  

We can prove an analogous statement in the total variation norm.
\begin{theorem}
  \label{th:total_variation}
  Let $\left\{A_N \right \}_{N > 1}$ be a sequence
  satisfying the same conditions as in theorem~{\rm
    \ref{th:main_theorem}}. We have
\begin{equation}
  \label{eq:total_variation_2}
  \int_{-\infty}^\infty\left \lvert f_N(x) -  \phi(x) \right \rvert
  dx = O_\epsilon\left(N^{-2  + b  + \epsilon}\right), \quad N \to \infty,
\end{equation}
where $\epsilon \in \left(0, \frac{1}{2} (1-b)\right)$.
\end{theorem}

As we discussed in the introduction, for technical reasons
theorems~\ref{th:main_theorem} and~\ref{th:total_variation} exclude
$b=1$.  Meckes'~\cite{Me08} result suggests that they are correct for
$b=1$ too.

The starting formula to prove theorems~\ref{th:main_theorem}
and~\ref{th:total_variation} is
\begin{equation}
  \label{eq:feller}
  e\left(N\right)\le \frac{1}{\pi} \int_{-T_N}^{T_N}\left \lvert
    \frac{\psi_N(\xi) - \psi(\xi)}{\xi}\right \rvert d\xi +
  \frac{24m}{\pi T_N}
\end{equation}
(see~\cite{Fel70}, p. 538), where 
\begin{equation}
  \label{eq:m_def}
   m := \max_{x \in \mathbb{R}}\abs{\phi(x)} =\frac{1}{\sqrt{2\pi}}
\end{equation}
and $T_N$ is an appropriate cut-off. Formula~\eqref{eq:feller}
transfers the problem of computing $e(N)$ into that of finding a bound
for $\abs{\psi_N\left(\xi\right) - \psi\left(\xi\right)}$
% \begin{equation}
%   \label{eq:char_fun_diff}
%    \abs{\psi_N\left(\xi\right) - \psi\left(\xi\right)}
% \end{equation}
for sufficiently large $\xi$.
 \begin{theorem}
\label{th:berry-esseen}
Let $C$ and $\delta$ be two constants independent of $N$ and let $0
\le \abs{\xi} <  \delta N^{(1-b)/2}$.  We have
\begin{equation}
  \label{eq:bound_1}
  \left \lvert \psi_N(\xi) - \psi(\xi) \right \rvert \le \frac{C\xi^4}%
  {N^{2-b}}e^{-\xi^2/2}.
\end{equation}
\end{theorem}
Throughout this paper $C$ will denote a constant, which may be
different at each occurrence.

\begin{remark}
\label{re:berry-essen}
Theorem~{\rm \ref{th:berry-esseen}} is of interest in its own
right. Such bounds are called \textbf{Berry-Ess\'een inequalities}.
They determine rates of convergence to central limit theorems, usually
for sums of independent or weakly dependent random variables. The
eigenvalues of random unitary matrices, however, exhibit a high degree
of correlation.

For eigenvalues of random unitary matrices, one consequence of such a
strong dependence is that the variance $\sigma^2 = \tr A_N A_N^*/(2N)$
remains finite in the limit $N \to \infty$. Instead, the variance of
the sum of $N$ independent random variables grows linearly in $N$.
When the moments diverge in the limit $N \to \infty$, just the first
few are enough to determine an optimal bound.  Since the right-hand
sides of equations~\eqref{eq:real_p} and~\eqref{eq:imag_p} converge to
normal random variables without any normalisation, the proof of
theorem~{\rm \ref{th:berry-esseen}} requires knowing the first $2N$
moments of $X_N$.
\end{remark}

\section{Moments and cumulants of $X_N$}
\label{sec:mom_cum}

The purpose of this section is to provide bounds and asymptotic
formulae for the moments and cumulants of $X_N$ that will be needed to
prove theorem~\ref{th:berry-esseen}. Most of these can be derived from
the results of Samuel~\cite{Sam80}, which we summarise in
\S\ref{sse:averages_mat_el}.

\subsection{Averages of matrix elements and the symmetric group}
\label{sse:averages_mat_el}

Samuel~\cite{Sam80} studied averages of the form
\begin{multline}
  \label{eq:samuel_int}
  \int_{\UN} U_{i_1 j_1} \dotsm U_{i_m j_m}\overline{U}_{k_1 l_1} \dotsm
  \overline{U}_{k_m l_m}d\mu_{\mathrm{H}}(U) \\ 
    = \sum_{\sigma, \tau \in
    \mathfrak{S}_m}M_{\sigma , \tau}(N)\delta_{i_1 k_{\sigma 1}}
  \dotsm \,
    \delta_{i_m k_{\sigma m}}\delta_{j_1 l_{\tau 1}} \dotsm \,
    \delta_{j_m l_{\tau m}},
\end{multline}
where $\mathfrak{S}_m$ denotes the symmetric group of degree $m$.  The
moments of $X_N$ are simply linear combinations of these integrals.

All the information on the averages~\eqref{eq:samuel_int} is contained
in the coefficients $M_{\sigma, \tau}(N)$.  A permutation of $m$
letters can always be factorised in a product of cycles.  It turns out
that  $M_{\sigma, \tau}(N)$ depends only on the cycle decomposition of
$\sigma \tau^{-1}$.

The lengths of the cycles of a permutation identify a sequence of
non-negative integers $\lambda = (\lambda_1,\dotsc,\lambda_k)$ such
that
\begin{equation}
  \label{eq:part_def}
  \lambda_1 \ge \dotsb\ge \lambda_k \quad \text{and} \quad 
   \abs{\lambda} := \lambda_1 + \dotsb + \lambda_k=m.
\end{equation}
In other words, there exists a one-to-one correspondence between the
cycle structures of $\mathfrak{S}_m$ and the set of
\textit{partitions} of $m$. The partition $\lambda(g)$ is called
\textit{cycle-type} of $g \in \mathfrak{S}_m$.  Therefore, we shall
adopt the notation
\begin{equation}
  \label{eq:partition_notation}
   M_\lambda(N) :=  M_{\sigma , \tau}(N),
\end{equation}
where $\lambda$ is the cycle-type of $\sigma  \tau^{-1}$.

A partition of $m$ is denoted by $\lambda \vdash m$; the addends
$\lambda_j$ are the \textit{parts} of $\lambda$.  An alternative
notation for a partition is the \textit{frequency representation}: if
$\lambda$ contains $r_1$ $1$s, $r_2$ $2$s and so forth, we write
$\lambda = \left(1^{r_1}\,2^{r_2}\dotsc\,m^{r_m}\right)$.  The length
of a partition $\ell(\lambda)$ is the largest $j$ such that $\lambda_j
>0$.  We also have
\begin{equation}
  \label{eq:length}
  \ell\left(\lambda\right)= r_1 + \dotsb + r_m.
\end{equation}
We shall find it convenient not to distinguish between two partitions
that differ only by a sequence of zeros at the end.  For example, 
 $(3,1,1)$ and $(3,1,1,0,0,0)$ are clearly the same partition.

Elements of $\mathfrak{S}_m$ that belong to the
same conjugacy class share the same cycle-type. Therefore, the
conjugacy classes of $\mathfrak{S}_m$ can be labelled by the set of
the partitions of $m$.  The number of elements in the
conjugacy class $\lambda$ is 
\begin{equation}
  \label{eq:g_lambda}
  g_\lambda: = \frac{m!}{1^{r_1} r_1! \dotsm
    m^{r_m} r_m!}.
\end{equation}
Furthermore, the conjugacy classes of $\mathfrak{S}_m$ are in
one-to-one correspondence with its irreducible representations, which
can be identified with the set of partitions of $m$ too. Since
characters are class functions they depend only on the cycle-types of
the permutations. The notation $\chi^\mu_\lambda$ indicates the
character of the irreducible representation $\mu$ evaluated on
elements of cycle-type $\lambda$.

Sometimes it is convenient to represent partitions using Young
tableaux. If $\lambda = (\lambda_1,\dotsc,\lambda_k)$, we draw $k$
left-justified rows of boxes, or nodes; the top row should contain
$\lambda_1$ boxes, the next one $\lambda_2$ and so on.  For example,
let $\lambda = (5,4,4,3,1)$.  Then,
\begin{equation*}
\begin{Young}
  & & & & \cr
  & & & \cr
  & & &\cr 
  & & \cr
 \cr
\end{Young}
\end{equation*}
is the corresponding Young tableau.

Samuel~\cite{Sam80} derived an explicit formula for 
$M_\lambda\left(N\right)$ when $m \le N$:
\begin{equation}
  \label{eq:cn_formula}
  M_\lambda\left(N\right) :=
  \frac{1}{m!}\sum_{\mu \vdash m}
  \frac{\dim V_\mu}%
  {f_\mu(N)} \, \chi^\mu_\lambda, \quad \lambda \vdash m, 
\end{equation}
where
\begin{equation}
  \label{eq:fr_def}
  f_\lambda(N) := \frac{1}{\dim V_\lambda}\sum_{\mu \vdash m}
  g_\mu\chi^\lambda_\mu N^{\ell\left(\mu\right)}
\end{equation}
and  
\begin{equation}
  \label{eq:hook_formula}
  \dim V_\lambda =      m! \frac{\prod_{1\le i < j\le \ell\left(\lambda\right)}
    \left(\lambda_i - \lambda_j +
      j-i\right)}{\prod_{j=1}^{\ell\left(\lambda\right)}
    \left(\lambda_j + \ell\left(\lambda\right) - j\right)}
\end{equation}
 is the dimension of the irreducible
representation $V_\lambda$.

The right-hand side of~\eqref{eq:fr_def} is polynomial in $N$ of
degree $m$.  It turns out that $f_\lambda(N)$ has only integer roots,
which have a simple representation in terms of the Young tableau of
$\lambda$; they are given by all the differences $i-j$, where $i$
counts the rows of the diagram in descending order and $j$ counts its
columns from left to right. For example, if $\lambda = (5,4,4,3,1)$,
then the roots of $f_\lambda(N)$ are
\begin{equation*}
  \begin{Young}
    0 & - 1 & -2 & - 3& -4 \cr
    1 & 0 & -1 & -2  \cr
    2 & 1 & 0 & -1 \cr
    3 & 2 & 1 \cr 
    4 \cr
  \end{Young}
\end{equation*}
We shall give a proof of this property later in this section.
% Samuel also conjectured an explicit
% formula for $f_\lambda(N)$:
% \begin{equation}
%   \label{eq:fr_def_2}
%   \begin{split}
%   f_\lambda(N) & = \prod_{j=1}^{\ell\left(\lambda\right)}\left(N + 1 -
%     j \right)
%   \left(N+2 -j\right) \dotsm \left(N + \lambda_j - j\right) \\
%   & = \prod_{j=1}^{\ell\left(\lambda\right)}\frac{\left(N + \lambda_j
%       - j\right)!}{\left(N-j\right)!}.
% \end{split}
% \end{equation}
% We give a proof of this conjecture in corollary~\ref{co:mom_1}.
% Let $\lambda'$ be the conjugate partition of $\lambda$ (see appendix A
% for the definition of $\lambda'$); we can show that 
% \begin{equation}
%   \label{eq:roots}
%   \ell(\lambda) - 1,\ell(\lambda) -2,\dotsc, 1, 0, - 1,
%   \dotsc,2- \ell(\lambda') - 2,1 -  \ell(\lambda'),
% \end{equation}
% are all the integer roots of the polynomial~\eqref{eq:fr_def}.  To
% complete the proof of~\eqref{eq:fr_def} one needs to count their
% multiplicity.

% \begin{equation}
%   \label{eq:fr_def}
%   f_\lambda(N) := \frac{1}{\chi^\lambda_{\lambda_e}}\sum_{\mu \vdash m}
%   g_\mu\chi^\lambda_\mu N^{\ell\left(\mu\right)}, 
% \end{equation}
% $\lambda_e := (1^m)$ is the cycle-type of the identity.
Since the characters of the irreducible representations of the
symmetric group are know via Frobenius's character formula,
equations~\eqref{eq:g_lambda}, \eqref{eq:cn_formula} and
\eqref{eq:fr_def} completely determine the
averages~\eqref{eq:samuel_int}.

Let $\lambda_1,\dots,\lambda_k$ be the parts of a partition $\lambda
\vdash m$. (We do not impose any ordering on the $\lambda_j$s.)  The
coefficients $M_\lambda (N)$ obey the recursion relations~\cite{Sam80}
\begin{equation}
  \label{eq:rec_rel_M}
  \begin{split}
  \delta_{\lambda_k
        1}M_{\left(\lambda_1,\dotsc,\lambda_{k-1}\right)}\left(N\right)
      & = 
  N M_{\left(\lambda_1,\dotsc,\lambda_k\right)}\left(N\right) 
  + \sum_{p + q = \lambda_k}
  M_{\left(\lambda_2,\dotsc,\lambda_k,p,q\right)}\left(N\right)\\
  & \quad  +  \sum_{j=1}^{k-1}
  \lambda_jM_{\left(\lambda_1,\dotsc,\lambda_{j-1},\lambda_j +
      \lambda_{k},
     \lambda_{j+1},\dotsc, \lambda_{k-1}\right)}\left(N\right), 
  \end{split}
\end{equation}
with initial condition $M_0(N)=1$. These equations do not depend on
permutations of the $\lambda_j$s and are a complete set, which
uniquely determines the coefficients $M_\lambda(N)$ for $\lambda \vdash
m$ in terms of those for $\lambda \vdash m-1$.

Traces of powers of matrices are homogeneous symmetric polynomials in
the eigenvalues.  Symmetric functions are intertwined with the
character theory of the symmetric group.  Therefore, it is not a
surprise that the formalism of symmetric polynomials will become
useful in computing the moments and cumulants of $X_N$. 

For every $j$ the \textit{power sum} of $m$ variables is
\begin{equation}
  \label{eq:power_part}
  p_j(x_1,\dotsc,x_m) := x_1^j + \dotsb + x_m^j.
\end{equation}
Next, we extend the definition~\eqref{eq:power_part} by taking the product
 \begin{equation}
   \label{eq:power_sym}
   p_\lambda :=p_{\lambda_1} \dotsm p_{\lambda_k} = \prod_{j=1}^mp_j^{r_j},
 \end{equation}
 where the $r_j$s are the frequencies of $\lambda$. Now suppose that
 $\ell(\lambda) \le m$. The \textit{Schur function}
 $s_\lambda(x_1,\dotsc,x_m)$ is defined by the ratio of two $m\times
 m$ determinants:
\begin{equation}
  \label{eq:schur_fun}
  s_\lambda(x_1,\dotsc,x_m):= \frac{\det\left(x_i^{\lambda_j + m
      - j}\right)}{\det\left(x_i^{m-j}\right)}.
\end{equation}
Schur functions are homogeneous symmetric polynomials of degree
$\abs{\lambda}$ and are related to the power sums by the formulae
(see~\cite{Mac95}, p.114)
\begin{equation}
  \label{eq:char_schr_p}
  s_\lambda = \frac{1}{m!}\sum_{\mu \vdash m} g_\mu \chi_\mu^\lambda \,
  p_\mu \quad \text{and} \quad  p_\mu = \sum_{\lambda \vdash m}
  \chi_\mu^\lambda s_\lambda.
\end{equation}
If $x_1,\dotsc,x_m$ are the eigenvalues of a $m \times m$ matrix $X$
we write $p_\lambda\left(X\right)=p_\lambda(x_1,\dotsc,x_m)$ and
$s_\lambda\left(X\right)=s_\lambda(x_1,\dotsc,x_m)$.

\begin{corollary}
  \label{co:proof_s_f}
  The polynomial~\eqref{eq:fr_def} can be factorised
  as\footnote{Samuel conjectured this formula in the appendix of his
    article~\cite{Sam80}, but did not give a proof.}
\begin{equation}
  \label{eq:fr_def_2}
  f_\lambda(N)  = \frac{1}{\dim V_\lambda}\sum_{\mu \vdash m}
  g_\mu\chi^\lambda_\mu N^{\ell\left(\mu\right)} =
   \prod_{(i,j)}\left(N - i + j\right) = 
   \prod_{j=1}^{\ell\left(\lambda\right)}\frac{\left(N + \lambda_j
      - j\right)!}{\left(N-j\right)!},
\end{equation}
where the pair $(i,j)$ span the row and column indices of the Young
tableau of $\lambda$.
\end{corollary}
\begin{proof}
  Let $\lambda \vdash m$ and $N$ be a positive integer, then
  $p_\lambda\left(I_N\right)= N^{\ell(\lambda)}$.  Therefore, from
  formula~\eqref{eq:char_schr_p} for $N \in \mathbb{Z}_{+}$
\begin{equation}
  \label{eq:schr_for_good}
  f_\lambda\left(N\right) = m! \frac{s_\lambda\left(I_N\right)}%
  {\dim  V_\lambda}. 
\end{equation}
The irreducible representations of the symmetric group and of
$\mathrm{GL}\left(N,\mathbb{C}\right)$ are related by the Schur-Weyl
duality.  If $N\ge \ell(\lambda)$ the Schur functions are precisely the
irreducible characters of $\mathrm{GL}(N,\mathbb{C})$.  Thus,
equation~\eqref{eq:schr_for_good} connects the dimensions of irreducible
representations of $\mathfrak{S}_m$ and
$\mathrm{GL}\left(N,\mathbb{C}\right)$ corresponding to the same
$\lambda$. Now, we have
\begin{equation}
  \label{eq:s_id}
  s_{\lambda}\left(I_N\right) = \prod_{1 \le j < k \le
    N}\frac{\lambda_j - \lambda_k + k -j}{k-j}.
\end{equation}
Combining this formula with~\eqref{eq:schr_for_good}
and~\eqref{eq:hook_formula} gives equation~\eqref{eq:fr_def_2}.
\end{proof}

% The ratio on the right-hand side of
% equation~\eqref{eq:f_for_good} is well known and is given by
% equation~\eqref{eq:fr_def_2} (see, \textit{e.g.}~\cite{FH91}, p. 78).

\subsection{The moments}
\label{sse:mom_calc}
Formula~\eqref{eq:iz_type_int} implies that $\psi_N$ is an entire
function.  Therefore, the series
\begin{equation}
  \label{eq:moments_expansion}
 \psi_N(\xi) = \sum_{n=0}^\infty \frac{\left(i\xi\right)^n}{n!} \mu_n
\end{equation}
converges in all the complex plane and defines all the moments of
$X_N$, which identify its probability distribution uniquely.  

Now, consider the Taylor expansion of the
integral~\eqref{eq:iz_type_int}:
\begin{equation}
  \label{eq:tayl_exp}
  \begin{split}
    \psi_N(\xi) & =
    \sum_{n=0}^\infty\frac{1}{n!}
    \left(\frac{i\xi}{2\sigma}\right)^n\sum_{m=0}^n\binom{n}{m}
    \int_{\UN}\left(\tr A_N U\right)^{n-m}\left(\tr A_N^*
      U^*\right)^{m}d\mu_{\mathrm{H}}(U).
  \end{split}
\end{equation}
Since Haar measure is left and right invariant, the integral in this
sum is zero unless $n=2m$.  Therefore,
\begin{equation}
  \label{eq:final_exp}
  \psi_N(\xi) = \sum_{m=0}^\infty  \frac{\left(-1\right)^m}{\left(m!\right)^2}
  \left(\frac{\xi}{2\sigma}\right)^{2m}I_N^m(A_N),
\end{equation}
where
\begin{equation}
  \label{eq:def_int_sam}
  I_N^{m}\left(A_N\right) := \int_{\UN}\abs{\tr A_N U}^{2m}d\mu_{\mathrm{H}}(U).
\end{equation}
Thus, the moments of $X_N$ are given by the formula
\begin{equation}
  \label{eq:moments_xN}
  \mu_{2m} = \frac{(2m-1)!!}{m!\left(2\sigma^2\right)^{m}}
  I_N^m\left(A_N\right).
\end{equation}
% Samuel~\cite{Sam80} computed an explicit expression for the
% integral~\eqref{eq:def_int_sam} up to $m\le N$ in terms of polynomials
% in the variables
% \begin{equation}
%   \label{eq:traces_1}
%  \tr A_N A_N^*,\dotsc,  \tr \left(A_NA_N^*\right)^m.
% \end{equation}

\begin{proposition}[Samuel 1980]
\label{le:my_mom}
Let $m\le N$ and let $\lambda =(1^{r_1}\dotsc \, m^{r_m})$ denote a partition
of $m$. We have
\begin{equation}
  \label{eq:samuel_formula}
  I_N^m(A_N) =  m! \sum_{\lambda \vdash m}
  g_\lambda M_\lambda\left(N\right)p_\lambda\left(A_NA_N^*\right).
\end{equation}
\end{proposition}
\begin{proof}
  The right-hand side of equation~\eqref{eq:samuel_int} can be
  re-written as
  \begin{equation}
    \label{eq:shifted_sum}
    \sum_{\rho,\tau \in \mathfrak{S}_m}M_\rho\left(N\right)\delta_{i_1
      k_{\rho  \tau 1}}
  \dotsm \,
    \delta_{i_m k_{\rho  \tau  m}}\delta_{j_1 l_{\tau 1}} \dotsm \,
    \delta_{j_m l_{\tau m}} ,
  \end{equation}
  where we have shifted the index in the sum by setting $\rho = \sigma
  \tau^{-1}$ and used the fact that $M_{\sigma,\tau}(N)$ depends only
  on $\sigma\tau^{-1}$.  By multiplying
  equation~\eqref{eq:shifted_sum} by $A_{N j_1i_1},\dotsc,A_{N j_m
    i_m}$ and $\overline{A}_{N l_1k_1},\dotsc,\overline{A}_{N l_m
    k_m}$ and summing over all indices, we obtain an expression of the
  form
\begin{equation}
  \label{eq:sum}
 \sum_{\rho  \in \mathfrak{S}_m}
  M_\rho\left(N\right) \sum_{\tau \in \mathfrak{S}_m}
     \sum_{\dotsc, \alpha, \beta,\gamma,\delta,\eta,\dotsc} \dotsm A_{N
    \alpha \beta}\overline{A}_{N\gamma \beta}A_{N \gamma \delta }
    \overline{A}_{N \eta \delta} \dotsm
\end{equation}
Consecutive indices in the inner sum, say $\beta$ and $\gamma$, are of
the type $k_{\rho \tau v}$ and $l_{\tau w}$ respectively, where $w =
\tau^{-1} \rho \tau v$. Thus, the collection of the addends such that
$v = \left(\tau^{-1}\rho \tau\right)^jv$ contributes with a factor
$\tr \left(A_N A_N^{*}\right)^j$.

Each letter belonging to a cycle of length $j$ is a fixed point of
order $j$ of every element in the conjugacy class of $\rho$.  The
inner sum in~\eqref{eq:sum} depends only on powers of $\tau^{-1} \rho
\tau$, and therefore is a class function and is independent of $\tau$.
Each cycle of length $j$ produces the factor $\tr
\left(A_NA_N^*\right)^j$. Therefore, we have
\begin{equation}
  \sum_{\dotsc, \alpha, \beta,\gamma,\delta,\eta,\dotsc} \dotsm A_{N
    \alpha \beta}\overline{A}_{N\gamma \beta}A_{N \gamma \delta }
    \overline{A}_{N \eta \delta} \dotsm = \prod_{j=1}^m \bigl(\tr \left(A_N
    A_N^*\right)^j\bigr)^{r_j} = p_\lambda\left(A_NA_N^*\right).
\end{equation}

Finally, formula~\eqref{eq:samuel_formula} follows from the fact that
$M_\rho\left(N\right)$ is a class function.
\end{proof}
\begin{remark}
  The integral~\eqref{eq:samuel_formula}, and thus
  by~\eqref{eq:moments_xN} the moments too, are linear combinations of
  the coefficients $M_\lambda(N)$, which have poles at the zeros of
  $f_\lambda(N)$. Such poles are related to certain singular integrals
  over $\UN$, which appear in lattice Quantum Chromodynamics and were
  first noted by De Wit and 't Hooft~\emph{\cite{DWtH77}}. They
  observed that such integrals are divergent for certain values of
  $N$.  The moments of $X_N$, however, are always finite.  The reason
  why the \textit{De Wit-'t Hooft anomalies} do not affect
  formula~\eqref{eq:samuel_formula} is because it is correct only for
  $m\le N$, and by corollary~\emph{\ref{co:proof_s_f}} the greatest
  zero of $f_\lambda(N)$ is $m-1$.
\end{remark}

As we mentioned at the beginning of this section, in order to prove
the Berry-Ess\'een inequality~\eqref{eq:bound_1} we need bounds and
asymptotic formulae for the moments and cumulants.  The evaluation of
the right-hand side of equation~\eqref{eq:samuel_formula} requires
Frobenius's character formula, which is quite cumbersome to use when
explicit formulae are needed. It turns out
that~\eqref{eq:samuel_formula} can be expressed in terms of
\text{Schur functions}, which allow it to be manipulated explicitly.
\begin{corollary}
\label{co:mom_1}
We have
\begin{equation}
  \label{eq:mom_schur}
  \mu_{2m}=\frac{(2m-1)!!}{\left(2\sigma^2\right)^m}\sum_{\lambda
    \vdash m}\left( \prod_{j=1}^{\ell\left(\lambda\right)}
    \frac{\left(N-j\right)!}{\left(N + \lambda_j  - j\right)!}\right) 
     \dim V_\lambda \, s_\lambda\left(A_NA_N^*\right).
\end{equation}
\end{corollary}
\begin{proof}
  From formulae~\eqref{eq:cn_formula}, \eqref{eq:char_schr_p} and
  \eqref{eq:samuel_formula} we obtain
\begin{equation}
  \label{eq:pass}
  \begin{split}
    I_N^m(A_N) & := m! \sum_{\lambda \vdash m} g_\lambda M_\lambda(N)
    p_\lambda\left(A_NA_N^*\right) \\
    & = \sum_{\lambda \vdash m} \sum_{\mu
      \vdash m} \frac{\dim V_\mu}{f_\mu(N)}g_\lambda \chi^\mu_\lambda
       \,p_\lambda\left(A_NA_N^*\right)\\
    & = m! \sum_{\mu \vdash
    m}\frac{\dim V_\mu }{f_\mu(N)}
  s_\mu\left(A_NA_N^*\right).
  \end{split}
\end{equation}

Equation~\eqref{eq:mom_schur} follows from formula~\eqref{eq:fr_def_2}.
\end{proof}

We are now in a position to find asymptotic formulae for the first $N$
moments of $X_N$.  Let us denote the moments of $\mathcal{N}(0,1)$ by
$\mun_{2m}$, \ie
\begin{equation}
  \label{eq:normal_mom}
  \mun_{2m} := \left(2m - 1\right)!!.
\end{equation}
\begin{proposition}
\label{pr:bounds_mom}
We have the following bounds:
\begin{align}
  \label{eq:mom_limit_38}
   \mu_{2m} & = \mun_{2m}\Bigl(1 +  O\left((m
     k)^mN^{-1}\right)\Bigr)\\
\intertext{and}
\label{eq:mom_limit_41}
\mu_{2m} & \le \left(N^{b}k\right)^m\mun_{2m}.
\end{align}
\end{proposition}
\begin{proof}
  From equations~\eqref{eq:cn_formula} and~\eqref{eq:samuel_formula}
  we have
\begin{equation}
  \label{eq:mom_bound}
  \mu_{2m} = \frac{\mun}{\left(2\sigma^2\right)^mm!}\sum_{\lambda
    \vdash m} \sum_{\mu \vdash m}\frac{\dim
    V_\mu}{f_\mu(N)} g_\lambda \chi^\mu_\lambda
  p_\lambda\left(A_NA_N^*\right). 
\end{equation}

The first step consists in finding bounds for
$p_\lambda\left(A_NA_N^*\right)$ and $f_\mu(N)$.  Remember that by
definition~\eqref{eq:bound}, the greatest singular value of $A_N$ is
bounded by $\sqrt{kN^b}$, where $k = O(1)$ and $0\le b < 1$.  We have
  \begin{equation}
    \label{eq:lead_order_asym}
    \tr \left(A_N A_N^*\right)^j = \left(2 \sigma^2\right)^j 
    N^{a + bj}c_j\left(A_N\right),  
  \end{equation}
  where $c_j\left(A_N\right) = O\left(k^j\right)$ and $a = 1-b$.
  Note that by definition $\sigma$ is independent of $A_N$, therefore
  $c_1\left(A_N\right)=1$. It follows that
  \begin{equation}
    \label{eq:asym_prod_traces}
   p_\lambda \left(A_N A_N^*\right)= \prod_{j=1}^m\bigl(\tr \left(A_N
    A_N^*\right)^j\bigr)^{r_j} = \left(2\sigma^2\right)^m
     N^{a \ell\left(\lambda\right) + bm}c_{\lambda}\left(A_N\right),
  \end{equation}
  where
  \begin{equation}
    \label{eq:con_part}
    c_{\lambda}\left(A_N\right) = \prod_{j=1}^m
    c_j\left(A_N\right)^{r_j} =  O\left(k^m\right).
  \end{equation}
  Denote by $\lambda_e=(1^m)$ the cycle-type of the identity in
  $\mathfrak{S}_m$. Combining equations~\eqref{eq:asym_prod_traces}
  and~\eqref{eq:con_part} we obtain.
\begin{equation}
  \label{eq:asym_prod_traces_2}
  p_\lambda\left(A_NA_N^*\right)  
  = \left(2 \sigma^2\right)^m 
  N^m \times 
   \begin{cases}
    1 & \text{if $\lambda = \lambda_e$,}\\
    O\bigl(k^m N^{-a\left(m-\ell\left(\lambda\right)\right)}\bigr)
     & \text{if $\lambda \neq \lambda_e$.}
   \end{cases}
\end{equation}

Now consider $f_\mu(N)$.  We can easily see that for $N \ge m + 1$
\begin{equation}
  \label{eq:bounds_fmu}
  f_{(1^m)}(N) = \prod_{j=1}^m (N - j) \le f_\mu (N) \le f_{(m^1)}(N)
  = \prod_{j=1}^m\left(N+j\right), \quad \mu \vdash m,
\end{equation}
where $(m^1)$ and $(1^m)$ correspond to the trivial and alternating
representations respectively, which are both one-dimensional.  We can
re-write the inequalities~\eqref{eq:bounds_fmu} in the following way
\begin{gather}
  \label{eq:fi}
  N^m\prod_{j=1}^m \left(1 - \frac{j}{N}\right)\le f_\mu(N) \le
  N^m\prod_{j=1}^m \left(1 +  \frac{j}{N}\right) \nonumber \\
  \frac{N^m m!}{\left(m + 1\right)^m} \le f_\mu\left(N\right) \le 
  \frac{N^m \left(2m + 1\right) \dotsb \left(m + 2\right)}{\left(m +
      1\right)^m}.  
\end{gather}

The sum~\eqref{eq:mom_bound} can be split as follows:
\begin{equation}
  \label{eq:split_1}
  \begin{split}
    \mu_{2m} & = \frac{\mun N^m}{m!} \sum_{\mu \vdash m} \frac{\left(\dim
      V_\lambda \right)^2}{f_\mu(N)} \\
   & \quad + \frac{\mun}{\left(2\sigma^2\right)^mm!}\sum_{\substack{\lambda
    \vdash m \\ \lambda \neq \lambda_e }} \sum_{\mu \vdash m}\frac{\dim
    V_\mu}{f_\mu(N)} g_\lambda \chi^\mu_\lambda
  p_\lambda\left(A_NA_N^*\right).  
  \end{split}
\end{equation}
The first sum on the right-hand side can be estimated using the
bounds~\eqref{eq:bounds_fmu} 
\begin{equation}
  \label{eq:first_sum}
   \begin{split}
  \frac{\mun N^m}{m!} \sum_{\mu \vdash m} \frac{\left(\dim
      V_\lambda \right)^2}{f_\mu(N)} & = \frac{\mun}{m!}\sum_{\mu \vdash
    m} \left(\dim V_\lambda\right)^2\bigl(1 +
  O\left(e^mN^{-1}\right)\bigr) \\
  & = \mun\bigl(1 + O\left(e^m N^{-1}\right)\bigr).
 \end{split}
\end{equation}
Using the same ideas, we write
\begin{equation}
  \label{eq:se_par}
\frac{\mun }{\left(2\sigma^2\right)^m m!}  
  \sum_{\mu \vdash m} \chi_\lambda^\mu \frac{\dim V_\lambda}{f_\mu(N)}
  = \frac{\mun }{\left(2\sigma^2\right)^m m! N^m}
  \sum_{\mu \vdash m}\chi_\lambda^\mu  \dim V_\lambda
  \bigl(1 + O\left(e^m N^{-1}\right)\bigr ). 
\end{equation}
Irreducible representations of finite groups can always be chosen to
be unitary.  Therefore, we have that $\abs{\chi^\mu_\lambda} \le
\chi^\mu_{\lambda_e} = \dim V_\mu$.  Thus, using the orthogonality of
the characters, the sum~\eqref{eq:se_par} becomes
\begin{equation}
  \label{eq:par_sum}
  \frac{\mun R(m,N)}{\left(2\sigma^2\right)^m N^m},
\end{equation}
where $ R(m,N)= O\left(e^m N^{-1}\right)$.
Finally, inserting equation~\eqref{eq:par_sum}
  into~\eqref{eq:split_1} and using~\eqref{eq:asym_prod_traces_2} we
  obtain formula~\eqref{eq:mom_limit_38}

  In order to prove equation~\eqref{eq:mom_limit_41}, recall that from
  formulae~\eqref{eq:mom_schur} and~\eqref{eq:fr_def_2} we can write
\begin{equation}
  \label{eq:mom_again}
  \mu_{2m} = \frac{\mun_{2m}}{\left(2\sigma^2\right)^m
    m!}\sum_{\lambda \vdash m} \left(\dim
    V_\lambda\right)^2\frac{s_\lambda\left(A_N
      A_N^*\right)}{s_\lambda\left(I_N\right)}.
\end{equation}
The greatest singular value of $A_N$ is bounded by $2\sigma^2
\sqrt{kN^b}$.  Since the eigenvalues of $A_N A_N^*$ are non-negative,
$s_\lambda\left(A_N A_N^*\right)$ is positive and
\begin{equation}
  \label{eq:bound_sh}
  s_\lambda\left(A_n A_N^*\right) \le \left(kN^b\right)^m 
  s_\lambda\left(I_N\right).
\end{equation}
Then, equation~\eqref{eq:mom_limit_41} follows from the orthogonality
of the characters.
\end{proof}

An immediate corollary is the convergence in distribution of $X_N$ to
$\mathcal{N}(0,1)$ (D'Aristotile \textit{et al.}~\cite{ADN03}). For
fixed $m$ formula~\eqref{eq:mom_limit_38} gives
\begin{equation}
  \label{eq:mom_lim}
  \lim_{N \to \infty} \mu_{2m} = \mun_{2m}.
\end{equation}

The bound~\eqref{eq:mom_limit_41} plays an important role in the proof
of the Berry-Ess\'een inequality~\eqref{eq:bound_1}.  For $N \gg
m$~\eqref{eq:mom_limit_38} is a better bound; however, it becomes much
worse when $m \approx N$. This is an important regime.  As we shall
see, when $b=1$ the right-hand side of~\eqref{eq:mom_limit_41} is too
large to allow~\eqref{eq:bound_1} to be valid for a range of $\xi$
sufficiently large for our purposes.  We believe that the correct
bound for $\mu_{2m}$ is much smaller than both~\eqref{eq:mom_limit_38}
and~\eqref{eq:mom_limit_41}.  The reason is that the sum
\begin{equation}
  \label{eq:sum_par_last}
   \sum_{\lambda \vdash m} \sum_{\mu \vdash m} \frac{\dim
     V_\mu}{f_\mu(N)}g_\lambda
     \chi_\lambda^\mu p_\lambda\left(A_NA_N^*\right)
\end{equation}
is characterised by a sequence of cancellations.

\begin{remark}
  It is worth noting that, since the integral on the right-hand side
  of equation~\eqref{eq:tayl_exp} is zero unless $n=2m$, the proof of
  proposition~{\rm \ref{pr:bounds_mom}}.  also demonstrates that the random
  variable $Z_N= X_N + i Y_N = \tr A_N U$ converges in distribution to
  a \textit{complex normal random variable} $Z$, whose centred
  moments\footnote{Here we have adopted the convention that the real
    and imaginary parts of a standard complex normal random variable
    have variance $\frac12$.  Therefore, if we had studied $\tr A_N
    U$, instead of its real and imaginary parts separately, we should
    have set $\sigma^2 = \tr A_N A_N^*/N$.  This explains the
    discrepancy of a factor $2^m$ in the notation used in
    equations~\eqref{eq:sam_mom} and~\eqref{eq:compl_mom}.} are
\begin{equation}
  \label{eq:compl_mom}
  \mathbb{E}\left\{Z^m\overline{Z}^n\right\} = \delta_{mn}\sigma^{2m}
  m!.
\end{equation}
\end{remark}

\begin{remark}
  \label{re:id_calc}
  When $A_N=I$ the first $N$ moment are exactly gaussian independently
  of $N$.  This is a particular case of a more general result proved
  by Diaconis and Shahshahani~{\rm\cite{DS94}} and can be easily
  recovered in our formalism. We have
   \begin{equation}
    \label{eq:sam_mom}
    \begin{split}
      I_N^m(I) & = m!\left(2\sigma^2\right)^m
       \sum_{\lambda \vdash m}
      g_\lambda M_\lambda\left(N\right) N^{\ell\left(\lambda\right)}\\
      & = \left(2\sigma^2\right)^m
       \sum_{\lambda \vdash m} \sum_{\mu \vdash m}
      \frac{\chi^\mu_{\lambda_e}\chi^{\mu}_{\lambda_e}}%
      {f_\mu\left(N\right)}\frac{g_\lambda\chi^\mu_\lambda}%
      {\chi^\mu_{\lambda_e}}N^{\ell\left(\lambda\right)}\\
      & = \left(2\sigma^2\right)^m\sum_{\mu \vdash m}
      \chi^\mu_{\lambda_e}\chi^\mu_{\lambda_e} =
      m!\left(2\sigma^2\right)^m, \quad m \le N.
  \end{split}  
  \end{equation}
\end{remark}

\subsection{The cumulants}
\label{sec:cumulants}

The characteristic function $\psi_N(\xi)$ is entire and by definition
$\psi_N(0)=1$.  Therefore, the Taylor series
\begin{equation}
  \label{eq:cumulant_def}
  \log \psi_N(\xi) =  \sum_{n=1}^{\infty} \frac{(i\xi)^n}{n!}\kappa_n
\end{equation}
converges in a neighbourhood of the origin. The coefficients
$\kappa_n$ are by definition the \textit{cumulants} of $X_N$ and
determine uniquely its probability distribution.  They are related to
the moments by the recurrence relation
\begin{equation}
  \label{eq:rec_rel}
  \kappa_n = \mu_n - 
   \sum_{k=1}^{n-1}\binom{n-1}{k-1}\kappa_k\mu_{n-k}.
\end{equation}

The choice of whether to use the moments or the cumulants depends on
the information that one is seeking to extract. It turns out that in
the proof of the Berry-Ess\'een inequality~\eqref{eq:bound_1} we shall
need the asymptotic behaviour of both. The purpose of this section is
to derive a bound for $\kappa_{2m}$ for $m \le N$.

Let $\lambda \vdash n$ and define
\begin{equation}
  \label{eq:partition_cumulant}
  \kappa_\lambda := \kappa_1 \dotsm \kappa_{\lambda_l} = 
  \prod_{j=1}^n\kappa_j^{r_j},
\end{equation}
where the $r_j$s are the frequencies of the partition $\lambda$.  There
exists an elegant formula (see, \textit{e.g.}~\cite{Mac95}, pp. 30--31)
that expresses the moments as polynomials in the cumulants:
\begin{equation}
  \label{eq:mom_cum}
   \mu_n = \sum_{\lambda \vdash n}c_\lambda \kappa_\lambda,
\end{equation}
where
\begin{equation}
  \label{eq:set_partition}
  c_\lambda := \frac{n!}{\left(1!\right)^{r_1} r_1!
    \dotsm \left(m!\right)^{r_m}r_m!}
\end{equation}
is the number of decompositions of a set of $n$ elements into disjoint
subsets containing $\lambda_1,\dotsc,\lambda_n$ elements.  Similarly,
equation~\eqref{eq:rec_rel} can be solved for the cumulants:
\begin{equation}
  \label{eq:late}
  \kappa_n = \sum_{\lambda \vdash
    n}\left(-1\right)^{\ell\left(\lambda\right)-1}\left(\ell\left(\lambda\right)
    -1\right)! c_\lambda \mu_\lambda, 
\end{equation}
where
\begin{equation}
  \label{eq:partition_moments}
  \mu_\lambda:= \mu_1 \cdots \mu_{\lambda_l} =   
  \prod_{j=1}^n\mu_j^{r_j}.
\end{equation}
 
All the odd moments of $X_N$ are zero, therefore all the odd
cumulants are zero too.  Thus, ~\eqref{eq:mom_cum} can be
rewritten as
\begin{equation}
  \label{eq:even_cumulants}
  \mu_{2m} = \left(2m - 1\right)!! \sum_{\lambda \vdash m}
  \frac{c_\lambda \kappa_{2\lambda}. }{\bigl(1!!\bigr)^{r_j}\dotsm 
    \bigl(\left(2m-1\right)!!\bigr)^{r_m}},
\end{equation}
where we have used the notation $2\lambda
=\left(2\lambda_1,\dotsc,2\lambda_l\right)$.

The $2m$-th moment of $X_N$ is a polynomial of degree $m$ in the
traces $\tr \left(A_N A_N^*\right)^j$; the recursion
relations~\eqref{eq:rec_rel} imply that the $2m$-th cumulant is also a
polynomial of degree $m$ in the same variables.  Therefore, we can
write
\begin{equation}
  \label{eq:cum_poly}
  \kappa_{2m}= \frac{\left(2m -1 \right)!!}{\left(2\sigma^2\right)^m} 
  \sum_{\lambda \vdash m} g_\lambda K_\lambda (N)
p_\lambda \left(A_N A_N^*\right).
\end{equation}
If we know the asymptotic behaviour of $K_\lambda(N)$, then we can
determine that of the cumulants.  In turn, the coefficients
$K_\lambda(N)$ are related to those of $M_\lambda(N)$.

The union $\lambda \cup \mu$ is defined as the partition whose parts
are those of $\lambda$ and $\mu$ arranged in descending order.
Cumulants have a combinatorial interpretation in term of partitions
of sets; let us define
\begin{equation}
  \label{eq:cum_M}
  M_\lambda (N) := \sum_{\Lambda}a_{\Lambda}
   \prod_{\mu \in \Lambda} K'_\mu(N),
\end{equation}
where $\Lambda$ runs through all possible distinct decompositions of
$\lambda$ as a union of sub-partitions.  The meaning of $a_{\Lambda}$
is better explained with an example.  Consider the partition $\lambda
= (\lambda_1,\lambda_2,\lambda_3)$ and write
\begin{equation}
  \label{eq:examp}
  \begin{split}
  M_{\left(\lambda_1,\lambda_2,\lambda_3\right)} & =
  K'_{\left(\lambda_1,\lambda_2,\lambda_3\right)} + 
  K'_{\left(\lambda_1\right)} K'_{\left(\lambda_2,\lambda_3\right)}  +
  K'_{\left(\lambda_2\right)} K'_{\left(\lambda_1,\lambda_3\right)} \\
 & \quad + K'_{\left(\lambda_3\right)} K'_{\left(\lambda_1,\lambda_2\right)} +
  K'_{\left(\lambda_1\right)}  K'_{\left(\lambda_2\right)}  
  K'_{\left(\lambda_3\right)}.
  \end{split}  
\end{equation}
If the $\lambda_1$, $\lambda_2$ and $\lambda_3$ are all
different, then each summand in~\eqref{eq:examp} is distinct, but if
some parts of $\lambda$ are repeated, this is not the case.  For
example, let $\lambda = (3,1,1)$, then $\left(\lambda_2\right) \cup
\left(\lambda_1,\lambda_3\right)$ and $\left(\lambda_3\right) \cup
\left(\lambda_1,\lambda_2\right)$ are the same decomposition of
$\lambda$ and
\begin{equation}
  \label{eq:ex_repetition}
  K'_{\left(\lambda_2\right)}K'_{\left(\lambda_1,\lambda_3\right)} =
  K'_{\left(\lambda_3\right)} K'_{\left(\lambda_1,\lambda_2\right)}.
\end{equation}
The coefficient $a_\Lambda$ is precisely such a
multiplicity. Computing it is an exercise in elementary combinatorics.

Let $\lambda \vdash m$ and define $\pi_\mu$ to be the number of times
that a partition $\mu$ appears in the decomposition $\lambda =
\bigcup_{\mu \in \Lambda} \mu$.  Furthermore, let $r_j$ and $s_j^\mu$
denote the frequencies of $j$ in $\lambda$ and $\mu$ respectively. We
have
\begin{equation}
  \label{eq:blambda}
  a_\Lambda = \prod_{j=1}^m r_j! \prod_{\mu \in \Lambda}
  \frac{1}{s^{\mu}_j!\pi_\mu!}.
\end{equation}

\begin{proposition}
  \label{le:combinatorics}
 The coefficients $K_\lambda(N)$ and $K^\prime_\lambda(N)$, 
  defined in equations~\eqref{eq:cum_poly}
  and~\eqref{eq:cum_M} respectively, coincide.
\end{proposition}
\begin{proof}
  For the sake of simplicity, let us set $x_j := \tr
  \left(A_NA_N^*\right)^j$ and $x_\lambda := x_1^{r_1} \dotsm
  x_m^{r_m}$.  By inserting equation~\eqref{eq:cum_poly} into the
  right-hand side of~\eqref{eq:even_cumulants} we see that 
  \begin{equation}
    \label{eq:lin_comb_1}
    \mu_{2m}=  \sum_{\lambda \vdash m} \sum_{\Lambda} b_\Lambda 
    \prod_{\mu \in \Lambda} K_{\mu}^{\pi_{\mu}}x_\mu^{\pi_\mu}.
  \end{equation}
  Similarly, by substituting~\eqref{eq:cum_M}
  into~\eqref{eq:samuel_formula} we obtain
  \begin{equation}
    \label{eq:lin_comb_2}
    \mu_{2m}=  \sum_{\lambda \vdash m} 
    \sum_{\Lambda} b^{\prime}_\Lambda \prod_{\mu \in \Lambda}
    K_{\mu}^{\pi_{\mu}}x_\mu^{\pi_\mu}.
  \end{equation}
  Since the right-hand sides of equations~\eqref{eq:lin_comb_1}
  and~\eqref{eq:lin_comb_2} identically equal for arbitrary $\mu$, we
  need to show that $b_\Lambda = b'_\Lambda$ 

  Equations~\eqref{eq:even_cumulants} and~\eqref{eq:cum_poly} give
  \begin{equation}
    \label{eq:first_passage}
    \begin{split}
      \mu_{2m} & = \frac{\left(2m -1\right)!!}{\left(2\sigma^2\right)^m}
     \sum_{\lambda \vdash m}c_\lambda \prod_{j=1}^m\left(\sum_{\mu \vdash j} g_{\mu}
        K_{\mu}(N)x_{\mu}\right)^{r_j} \\
      & = \frac{\left(2m - 1\right)!!}{\left(2\sigma^2\right)^m}
       \sum_{\lambda \vdash m} c_\lambda
      \prod_{j=1}^m r_j!
      \sum_{\substack{\pi_\mu \\  \sum\limits_{\mu} \pi_\mu = r_j }}
      \prod_{\mu \vdash j} \frac{g_{\mu}^{\pi_{\mu}}}{\pi_{\mu}!}
      K^{\pi_{\mu}}_{\mu}(N)x^{\pi_{\mu}}_{\mu}\\
    & = \frac{\left(2m - 1\right)!!\, m!}{\left(2\sigma^2\right)^m}
     \sum_{\lambda \vdash m} \frac{1}%
       {\left(1!\right)^{r_1}\dotsm\left(j!\right)^{r_j}}
      \sum_{\Lambda} \prod_{\mu \in \Lambda}
  \frac{g_\mu^{\pi_\mu}}{\pi_\mu!}  
  K^{\pi_{\mu}}_{\mu}(N)x^{\pi_{\mu}}_{\mu}.
\end{split}
\end{equation}
In the last passage $\pi_\mu$ assumes the same meaning as in
equation~\eqref{eq:blambda}, \ie it is the number of repetitions of a
partition $\mu$ in the union $\lambda = \bigcup_{\mu \in \Lambda}
\mu$. Now, let $\lambda =(1^{r_1}\dotsc \, m^{r_m})$ and 
$\mu = (1^{s^\mu_1} \dotsc \, j^{s^\mu_j})$ with $\lambda \vdash m$ and
$\mu \vdash j$.  The frequencies of $\lambda$ and $\mu$ are related by 
\begin{equation}
  \label{eq:rep_rel}
  r_k = \sum_{\mu \in \Lambda} s^{\mu}_k.
\end{equation}
Furthermore, by definition we have
\begin{equation}
  \label{eq:g_again}
  g_\mu = \frac{j!}{1^{s^\mu_1}s^\mu_1!\dotsm j^{s^\mu_j}s^\mu_j!}.
\end{equation}
Thus, combining equations~\eqref{eq:first_passage}, \eqref{eq:rep_rel}
and~\eqref{eq:g_again} we arrive at
\begin{equation}
  \label{eq:final_first_1}
  \mu_{2m} = \frac{\left(2m\right)!}{\left(2\sigma\right)^{2m}}
  \sum_{\lambda \vdash m} \frac{1}{1^{r_1}\dotsm m ^{r_m}}
  \sum_{\Lambda} \prod_{j=1}^m \prod_{\mu \in \Lambda} 
  \frac{1}{s_j^\mu!\pi_\mu!}\prod_{\mu \in \Lambda}
    K^{\pi_{\mu}}_{\mu}(N)x^{\pi_{\mu}}_{\mu}.
\end{equation}
Finally, equations~\eqref{eq:moments_xN},
\eqref{eq:samuel_formula}, \eqref{eq:cum_M} and~\eqref{eq:blambda}
give  
\begin{equation}
  \label{eq:final_first_2}
  \mu_{2m} = \frac{\left(2m\right)!}{\left(2\sigma\right)^{2m}}
  \sum_{\lambda \vdash m} \frac{1}{1^{r_1}\dotsm m ^{r_m}}
  \sum_{\Lambda} \prod_{j=1}^m \prod_{\mu \in \Lambda} 
  \frac{1}{s_j^\mu!\pi_\mu!}\prod_{\mu \in \Lambda}
    K^{\prime \pi_{\mu}}_{\mu}(N)x^{\pi_{\mu}}_{\mu}.
\end{equation}
\end{proof}

Brouwer and Beenaker~\cite{BB96} computed the leading order
asymptotics as $N \to \infty$ of $K_\lambda(N)$. By inserting the
right-hand side of~\eqref{eq:cum_M} in equations~\eqref{eq:rec_rel_M}
we derive the recursion relations
\begin{multline}
  \label{eq:rec_rel_K}
  N K_{\left(\lambda_1,\dotsc,\lambda_k\right)}\left(N\right) 
  + \sum_{p + q = \lambda_k}
  K_{\left(\lambda_2,\dotsc,\lambda_k,p,q\right)}\left(N\right) +
  \sum_{j=1}^{k-1}
  \lambda_jK_{\left(\lambda_1,\dotsc,\lambda_{j-1},\lambda_j +
      \lambda_{k},
     \lambda_{j+1},\dotsc, \lambda_{k-1}\right)}\left(N\right) \\
      + \sum_{p+q=\lambda_k} \sum_{j=1}^{k-1}\frac{1}{j!\left(k - j -
          1\right)!} \sum_{\sigma \in
        \mathfrak{S}_{k-1}}K_{\left(\sigma 1,\dotsc,\sigma
          j,p\right)}(N)K_{\left(\sigma (j+1),\dotsc,\sigma (k-1),q\right)}(N)=0
\end{multline}
with $K_0\left(N\right)=1$.
The solution to these equations to leading order is
\begin{equation}
  \label{eq:bb_exp}
  \begin{split}
  K_\lambda(N)& =\left(-1\right)^{m
    +\ell\left(\lambda\right)} 2^{\ell\left(\lambda\right)} N^{-2m -
    \ell\left(\lambda\right) +2} \\
    & \quad \times  \frac{\left(2m +
      \ell\left(\lambda\right) -
      3\right)!}{\left(2m\right)!}\prod_{j=1}^m \frac{\bigl(\left(2j
      -1\right)!\bigr)^{r_j}}{\bigl(\left(j - 1\right)!\bigr)^{2r_j}}
     + O\left(N^{-2m - \ell\left(\lambda\right)}\right).
   \end{split}
\end{equation}

We are now in a position to state the main result of this section.
\begin{theorem}
  \label{th:as_cum}
  We have
  \begin{equation}
    \label{eq:th_cum_asy}
    \kappa_{2m} = O\left((2m)! N^{-(2-b)(m -1) }\right).
  \end{equation}
\end{theorem}
\begin{proof}
This bound follows simply by combining equations~\eqref{eq:bb_exp},
\eqref{eq:cum_poly} and~\eqref{eq:asym_prod_traces_2}.  
\end{proof}

\section{Proof  of  the Berry-Ess\'een inequality}
\label{sec:proof_be}

In order to prove the Berry-Ess\'een bound~\eqref{eq:bound_1}, we need
an estimate of the radius of convergence of the cumulant
expansion~\eqref{eq:cumulant_def}.  

\begin{lemma}
  \label{lem:zer_psi}
  There exists a constant $\delta >0$ such that $\psi_N(\xi) >0$ for $
  0 \le \abs{\xi} \le \delta N^{\frac{1-b}{2}}$.
\end{lemma}
\begin{proof}
  Since $\psi_N(\xi)$ is entire, the radius of convergence of the
  Taylor series of $\log \psi_N(\xi)$ is given by the location of the
  nearest zero to the origin of $\psi_N(\xi)$.   

  By definition
  \begin{equation}
    \label{eq:bound_four}
    \abs{\psi_N(\xi)} \le \psi_N(0) =\int_{-\infty}^\infty f_N(x)dx =1.
  \end{equation}
  Suppose that $\psi(\xi)$ has real zeros and let $\bar{\xi}$ be the
  closest to the origin. Since $\psi_N(\xi)$ is even, we can assume
  that $\bar{\xi}$ is positive. For $ \abs{\xi}< \bar{\xi}$, $0 <
  \psi_N(\xi) \le 1$, therefore the Taylor series of $\log
  \psi_N(\xi)$ is convergent in $\left(-\bar{\xi},\bar{\xi}\right)$.
  Thus, it also converges in a circle centred at the origin and of
  radius $\bar{\xi}$.  In other words, there are not any complex zero
  of $\psi$ whose distance from the origin is less than $\bar{\xi}$.
  Therefore, in the rest of this proof we can take $\xi$ to be real
  and positive.

A general formula (see ~\cite{Fel70}, p. 514) for moment
generating functions gives
\begin{equation}
  \label{eq:first_in}
  \abs{\psi_{N}\left(\xi\right) - \sum_{j=0}^{k-1}\left(-1\right)^j
    \frac{\xi^{2j}}{\left(2j\right)!}\mu_{2j}}
      \le \frac{\xi^{2k}}{\left(2k\right)!}\mu_{2k}.
\end{equation}
Let us consider the two sums
\begin{subequations}
\label{eq:sums}
\begin{align}
  \label{eq:odd}
  v_{2r-1}(\xi)  & := \sum_{j=0}^{2r -  1} (-1)^j
  \frac{\xi^{2j}}{\left(2j\right)!}\mu_{2j},\\
  \label{eq:even}
  v_{2s}(\xi) & := \sum_{j=0}^{2s} (-1)^j
  \frac{\xi^{2j}}{\left(2j\right)!}\mu_{2j}.
\end{align}
\end{subequations}
Since the Taylor expansion of $\psi_N(\xi)$ is an alternating series,
equation~\eqref{eq:first_in} implies
\begin{equation}
  \label{eq:bound_psi}
  v_{2r - 1}(\xi) \le \psi_N(\xi) \le v_{2s}(\xi)
\end{equation}
for any pair of integers $r \ge 1$ and $s\ge 0$. By definition $\mu_2
=1$, thus the lemma is trivially true for $\xi^2 < 2$.

Let us write
\begin{equation}
  \label{eq:exp_dec}
  \exp\left(-\omega \xi^2/2\right) = w_{2r-1}\left(\xi\right) +
    u_{2r}\left(\xi\right), 
\end{equation}
where
\begin{subequations}
\begin{align}
  \label{eq:exp_app}
  w_{2r - 1}(\xi) := \sum_{j=0}^{2r - 1} (-1)^j
  \frac{\omega^j \xi^{2j}}{\left(2j\right)!} \mun_{2j},\\
  \label{eq:reminder_exp}
  u_{2r}\left(\xi\right) := \sum_{j=2r}^\infty \left(-1\right)^j
  \frac{\omega^j \xi^{2j}}{\left(2j\right)!} \mun_{2j}.
\end{align}
\end{subequations}
Recall that $\mun_{2j}=\left(2j -1\right)!!$ denotes the moments of
$\mathcal{N}(0,1)$. We choose $\omega > 4e^2$ and independent of $r$.
We now want to show that for $r \le N$ there exists an appropriate
$\omega$ such that
\begin{equation}
  \label{eq:ine_essent}
  0 < w_{2r - 1}(\xi) \le v_{2r-1}(\xi)
\end{equation}
in the interval
\begin{equation}
   \label{eq:interval}
    \sqrt{2} \le \frac{2e \left(r-1\right)^{1/2}}{\omega^{1/2}}
     \le \xi <\frac{2e r^{1/2}}{\omega^{1/2}}.
\end{equation}

Since $\omega \xi^2 <4e^2r$ the summands in the
reminder~\eqref{eq:reminder_exp} are strictly decreasing.  Therefore,
we can write
\begin{equation}
  \label{eq:first_ineq}
  \begin{split}
   w_{2r - 1}(\xi) & = \exp\left(-\omega \xi^2/2\right) -  
     \sum_{j=2r}^\infty 
    \left(-1\right)^j \frac{\omega^j
    \xi^{2j}}{\left(2j\right)!} \mun_{2j} \\
     & > \exp\left(-\omega
    \xi^2/2\right) - \frac{\omega^{2r}
    \xi^{4r}}{\left(4r\right)!} \mun_{4r} >0.
  \end{split}
\end{equation}
The last passage is a straightforward consequence of Stirling's
formula. 

Now, both $w_{2r-1}(\xi)$ and $v_{2r-1}(\xi)$ are alternating sums.
Therefore, $w_{2r-1}(\xi) \le v_{2r-1}(\xi)$ if
\begin{equation}
  \label{eq:ineq_mom}
  \frac{\mun_{2(2j)}}{\left(4j\right)!}\omega^{2j}\xi^{4j}  -
  \frac{\mun_{2(2j +1)}}{\left(4j+2\right)!}\omega^{2j+1}\xi^{4j+2}
  \le \frac{\mu_{2(2j)}}{\left(4j\right)!}\xi^{4j}  -
  \frac{\mu_{2(2j +1)}}{\left(4j+2\right)!}\xi^{4j+2}
\end{equation}
for $j \le r - 1$.  This equation can be rearranged as follows
\begin{equation}
  \label{eq:ineq_mom_2}
  \left(\frac{\xi^2}{2}\right)^{2j}\left(\omega^{2j}  -
  \frac{\mu_{2\left(2j\right)}}{\mun_{2\left(2j\right)}}\right)
   \le   \frac{1}{\left(2j + 1 \right)}
   \left(\frac{\xi^2}{2}\right)^{2j+ 1}\left(\omega^{2j + 1}  -
   \frac{\mu_{2\left(2j + 1\right)}}{\mun_{2\left(2j + 1\right)}}\right).
\end{equation}
If we choose $\omega > N^bk$, this inequality holds for $r \le N$
because of proposition~\ref{pr:bounds_mom} and
equation~\eqref{eq:mom_limit_38}.  Thus, the statement of the lemma
follows if we set $\delta = 2e/\left(N^b\omega\right)^{1/2}$.
\end{proof}

We are now in a position to prove theorem~\ref{th:berry-esseen}. From
theorem~\ref{th:as_cum}, we know that for $m \le N$
\begin{equation}
  \label{eq:order_cumu_N}
  \kappa_{2m} = O\left((2m)! N^{-(2-b)(m -1) }\right).
\end{equation}

Furthermore, from formulae~\eqref{eq:moments_xN}
and~\eqref{eq:rec_rel} it is straightforward to compute the first few
cumulants.  We have
\begin{subequations}
\label{eq:cum_2}
\begin{align}
  \kappa_2 & = 1, \\
  \label{eq:4th_cumulant}
  \kappa_4 & = - \frac{3\tr\left(A_NA_N^* -2\sigma I_N\right)^2}%
  {4 \sigma^2N^3\left(1 - 1/N^2\right)},\\
  \kappa_6 &= \frac{15 \tr\left(A_NA_N^* -2\sigma I_N\right)^3}{2N^5
    \sigma^3\left(1 - 1/N^2\right)\left(1 - 4/N^2\right)}.
\end{align}
\end{subequations}

Since the cumulant expansion converges up to $\xi \le \delta
N^{\frac{1}{2}\left(1-b\right)}$, there exists a parameter $\theta$
such that
\begin{equation}
  \label{eq:log_eq}
  \log \psi_{N}\left(\xi\right) = - \frac{\xi^2}{2}  +
  \theta\frac{\xi^4}{N^{2-b}}. 
\end{equation}
It turns out that $\theta= O(1)$ as $N \to \infty$. Now, recall that
the moment generating function of $\mathcal{N}(0,1)$ is $\psi(\xi) =
e^{-\xi^2/2}$.  Therefore, we can write
\begin{equation}
  \label{eq:final_passages}
  \abs{\psi_N\left(\xi\right) - \psi(\xi)} = 
   e^{-\xi^2/2}\abs{e^{\theta \xi^4/N^{(2-b)}} - 1} \le \frac{\theta
    \xi^4}{N^{2-b}} e^{\theta \xi^4 /N^{(2-b)}}e^{-\xi^2/2},  
\end{equation}
where we have used the inequality $\abs{e^z - 1}\le
\abs{z}e^{\abs{z}}$.  The exponential $e^{\theta \xi^4 /N^{2-b}}$ is
bounded in $N$ provided $\theta = O(1)$.  Therefore, the right-hand
side of~\eqref{eq:final_passages} becomes
\begin{equation}
  \label{eq:almost_there}
  \abs{\psi_N\left(\xi\right) - \psi(\xi)} 
    \le \frac{C\xi^4}{N^{2-b}}e^{-\xi^2/2},
\end{equation}
where $C$ can be chosen independent of $N$.

To complete the proof of equation~\eqref{eq:almost_there}, we need to
show that if $\xi \le \delta N^{\frac{1}{2}(1-b)}$, then $\theta =
O(1)$.  Let us write the cumulant expansion as
\begin{equation}
  \label{eq:cumulant_rem}
  \log \psi_N(\xi) = -\frac{\xi^2}{2} + \frac{\kappa_4 \xi^4}{4!} + R_6(N),
\end{equation}
where
\begin{equation}
  \label{eq:reminder}
  R_6(N) = \sum_{m=3}^\infty \frac{\kappa_{2m}\xi^{2m}}{\left(2m\right)!}.
\end{equation}
If a series $\sum_{m=1}^\infty c_m$ converges , then $c_m \to
0$ as $m \to \infty$.  Therefore, for $m > N$ we must have.
\begin{equation}
  \label{eq:order_cum_gN}
  \kappa_{2m} = o\Bigl((2m)!\left(\delta^2 N\right)^{-m(1-b)}\Bigr),
  \quad m \to \infty.
\end{equation}
Thus, combining equations~\eqref{eq:order_cumu_N}
and~\eqref{eq:order_cum_gN}, the reminder~\eqref{eq:reminder} can be
bound by the series
\begin{equation}
  \label{eq:bound_series}
   \frac{C_1 \xi^6}{N^{2(2-b)}}\sum_{m=0}^{N-3}
   \frac{\xi^{2m}}{N^{(2-b)m}}  +
    \frac{C_2\xi^{2(N+1)}}{N^{(1-b)(N+1)}}\sum_{m=0}^\infty
   \frac{\xi^{2m}}{\left(\delta^*N^{(1-b)}\right)^{m}},
\end{equation}
where $C_1$ and $C_2$ are constants and $\delta^* > \delta$. For $\xi
< N^{\frac12\left(1-b\right)}$ this sum is $O(1)$ as $N \to \infty$,
which implies that $\theta$ cannot be an increasing function of $N$.

\begin{remark}
  There is striking difference between the superexponential rate of
  convergence discovered by Johansson~{\rm \cite{Joh97}} when $A_N$ is
  the identity and the rates of Theorem~{\rm \ref{th:main_theorem}}.
  Indeed, superexponential rates of convergence to central limit
  theorems are unusual in probability theory.  Theorem~{\rm
    \ref{th:berry-esseen}} provides some insight into this. When
  $A_N=I_N$ the first $N$ moments of $X_N$ are gaussian (see
  equation~\eqref{eq:sam_mom}) and its first $N$ cumulants but
  $\kappa_2$ are zero. Therefore, equation~\eqref{eq:almost_there}
  turns into
  \begin{equation}
  \label{eq:bound_2}
    \left \lvert \psi_N(\xi) - \psi(\xi) \right \rvert \le
    \frac{C\xi^{2(N+1)}}{N^{N+1}}e^{-\xi^2/2}.
  \end{equation}
\end{remark}

\section{Proof of theorem~\ref{th:main_theorem}}
\label{sec:proof_main_th}

\subsection{Preliminaries}
\label{sssec:full_rank}
Let us set $S_N = \delta N^{(1-b)/2}$ and $T_N= N^\gamma$, where
$\gamma >2$.  Theorem~\ref{th:berry-esseen} allows us to split the
right-hand side of~\eqref{eq:feller} as follows:
\begin{equation}
  \label{eq:split}
  \begin{split}
  e(N) & \le \frac{2C}{N^{2-b}\pi}\int_0^{S_N}\xi^3
  e^{-\xi^2/2}d\xi + \frac{2}{\pi S_N  }\int_{S_N}^{N^\gamma}
  e^{-\xi^2/2}d\xi + \frac{2}{\pi S_N  }\int_{S_N}^{N^\gamma}\left
    \lvert \psi_N(\xi)\right \rvert d\xi \\
    & \quad  + \frac{24}{\sqrt{2\pi^3}N^\gamma}.
  \end{split}
\end{equation}
The upper limits of integration can be replaced by infinity.  The
first integral gives the desired bound.  We need to show the remaining
terms are of lower order.

The second integral in equation~\eqref{eq:split} can be rewritten as
\begin{equation}
\label{eq:def_err}
\int_{\delta N^{(1-b)/2}}^{N^{\gamma}} e^{-\xi^2/2}d\xi \le
\frac{1}{2}\sqrt{\frac{\pi}{2}} \erfc\left(\delta N^{(1-b)/2}/\sqrt{2}
\right),
\end{equation}
where
\begin{equation}
  \label{eq:error_fun}
  \erfc(t) := \frac{2}{\sqrt{2}}\int_t^\infty e^{-x^2}dx
\end{equation}
is the complementary error function. Since $\erfc(t)$ satisfies the
inequalities (see, \eg\cite{AS72}, p. 298)
\begin{equation}
  \label{eq:ineq_er_fun}
  \frac{1}{t + \sqrt{t^2 + 2}} < e^{t^2}\int_t^\infty e^{-\xi^2}d\xi
  \le \frac{1}{t + \sqrt{t^2 + \frac{4}{\pi}}},
\end{equation}
the second integral in~\eqref{eq:split} can be neglected.

The last task that we are left with is to estimate the integral
\begin{equation}
  \label{eq:third_int}
  \int_{S_N}^{N^\gamma} \abs{\psi_N(\xi)}d\xi = O\left(N^{-2}\right).
\end{equation}

\subsection{Regularity properties of the distribution of $X_N$}
\label{sse:reg_prop}

In general we do not have an explicit formula for $\psi_N(\xi)$ in the
interval $(\delta N^{(1-b)/2}, N^\gamma)$. Thus, in order to
estimate its behaviour in this range we need to adopt an indirect
approach.  The idea is to approximate $X_N$ with a random variable
$X_N^*$ whose characteristic function allows us to control the third
integral in equation~\eqref{eq:split}.  Then, we will estimate the
difference between $e(N)$ and
\begin{equation}
  \label{eq:new_error}
  e^*(N) := \sup_{x \in \mathbb{R}} \abs{F^*_N(x) - \Phi^*(x)},
\end{equation}
where $F^*_N$ is the approximate distribution function of $X_N^*$ and
$\Phi^*$ is the distribution function of a random variable close to
$\mathcal{N}(0,1)$ (in a sense that will be made precise later).

We first need to discuss some regularity properties of the probability
distribution of $X_N$.\footnote{In section~\ref{sse:sing_deg} the
  ability of estimating a bound for $\abs{f_N(x) - \phi(x)}$ will be
  essential.  Even though $\UN$ is compact and Haar measure is absolutely
  continuous with respect to the Lebesgue measure on
  $\mathbb{R}^{N^2}$, it is far from obvious that $f_N$ is bounded or
  even continuous in all $\mathbb{R}$.  For example,
  lemma~\ref{le:absolutely_continuity} is false for $N=1$.  Indeed,
  a direct calculation gives $f_1(x)=\bigl(1 -
    \left(x/a\right)^2\bigr)^{-1/2}$.}
\begin{lemma}
  \label{le:absolutely_continuity}
  If $N > 1$ the distribution function $F_N$ is absolutely
  continuous, it admits the integral representation
   \begin{equation}
     \label{eq:prob_dis}
     F_N(x) = \int_{-\infty}^x f_N(t)dt,
   \end{equation}
   where $f_N \in L^1\left(\mathbb{R}\right)$, is bounded and
   uniformly continuous.  Furthermore,
   $\psi_N(\xi) = o\left(\xi^{-1}\right)$ as $\xi \to \infty$.
\end{lemma}
\begin{proof}
Denote by $\Theta(N)$ the maximal torus of $\UN$, \ie the group of
diagonal unitary matrices
\begin{equation}
  \label{eq:diag}
  \diag\left(e^{i\theta_1},\dotsc,e^{i\theta_N}\right) = 
  \begin{pmatrix}
           e^{i\theta_1} & & \\
               & \ddots & \\
              & & e^{i\theta_N}
        \end{pmatrix}, 
\end{equation}
and write $\mathrm{W}\left(N\right) = \UN /\Theta(N)$. An explicit
expression for Haar measure on
$\UN$ is 
\begin{equation}
  \label{eq:haar}
  d\mu_{\mathrm{H}}\left(U\right) =  \frac{2^{N\left(N-1\right)/2}}%
   {(2\pi)^N N!} \prod_{1\le j < k \le N} \sin^2\left(\frac{\theta_k -
        \theta_j}{2}\right)d\theta_1 \dotsc d\theta_N d\mu_{\mathrm{W}},
\end{equation}
where $d\mu_{\mathrm{W}}$ is a normalized Borel measure on
$\mathrm{W}\left(N\right)$. 

Now, recall that
\begin{equation}
  \label{eq:XN_again}
  X_N = a_1\cos\theta_1 + \dotsb + a_N\cos \theta_N,
\end{equation}
where $a_1,\dotsc,a_N$ are the diagonal elements of $A_N$.  Thus, we
can integrate out $d\mu_{\mathrm{W}}$ and study the measure
\begin{equation}
  \label{eq:integrated_out}
  d\mu_{\Theta} = \frac{2^{N\left(N-1\right)/2}}%
   {(2\pi)^N N!} \prod_{1\le j < k \le N} \sin^2\left(\frac{\theta_k -
        \theta_j}{2}\right)
  d\theta_1 \dotsm d\theta_N.
\end{equation}
Since $X_N$ is an absolutely continuous function of
$\theta_1,\dotsc,\theta_N$, if $D \subset \left [0,2\pi \right)^N$ is
a set whose image $X_N(D)$ has Lebesgue measure zero, then $D$ must
have zero measure too. It follows from
equation~\eqref{eq:integrated_out} that $\mathbb{P}\left\{ X_N \in
  B\right\}=0$ for any set $B$ of Lebesgue measure zero. Therefore the
probability distribution of $X_N$ is absolutely continuous. Since the
only absolutely continuous measures on $\mathbb{R}$ are only those
that have a density, $F_N$ admits the integral
representation~\eqref{eq:prob_dis}.

We can say more about $f_N$.  The measure $e^{i \xi X_N}d\mu_\Theta$
is a differential form on the $N$-dimensional torus.  Let
\begin{equation}
  \alpha_N = \frac{1}{\sigma}\max_{U \in \UN} \abs{\tr A_N U } = a_1 + \dotsb + a_N.
\end{equation}
For any neighborhood of §$X_N \in [-\alpha_N,\alpha_N]$ we can find a
local change of variables that allows us to write
\begin{equation}
  \label{eq:exterior_prod}
e^{i\xi X_N}  d\mu_\Theta = e^{i \xi X_N} \ed X_N \wedge \omega,
\end{equation}
where $\omega$ is $(N-1)$-form on $\Theta(N)$, the symbol $\wedge$
denotes the exterior product and the roman `$\ed$' indicates exterior
differentiation.\footnote{While an absolute continuous measure $d\mu$
  on a smooth manifold can always be interpreted as a differential
  form, it does not mean that it is the exterior derivative of another
  form.  Indeed, $d\mu_\Theta$ is not. We use the notation `$\ed$' to
  emphasise this difference, because it is important in what follows.}
For example, we can choose
\begin{equation}
  \label{eq:change_var}
  \begin{split}
    x_N & =  a_1\cos \theta_1 + a_2\cos\theta_2 +\dotsb + a_N \cos \theta_N \\
    \varphi_2 &= a_1 \sin \left(\theta_1 +\beta\right) +
    a_2\sin\left(\theta_2 - \beta\right)  \\
    \varphi_3 & =\theta_3 \\
    \vdots \:  \: & \quad \: \:\vdots \\
    \varphi_N & = \theta_N,
 \end{split}
\end{equation}
where $\beta$ is a real parameter.  The Jacobian of this transformation is
\begin{equation}
  \label{eq:jacobian}
  \begin{split}
  J(x_N,\phi_j) &=    \frac{\partial \left(x_N,\varphi_j\right)}%
    {\partial \left(\theta_1,\ldots,\theta_N\right)}  = a_1a_2\cos \beta
    \sin\bigl(\theta_2\left(x_N,\varphi_j\right)
    -\theta_1\left(x_N,\varphi_j\right)\bigr) \\
    & \quad + \frac{a_1a_2\sin \beta}{2}\cos\bigl(\theta_1(x_N,\phi_j)
      + \theta_2(x_N,\phi_j)\bigr) \\
    & \quad +\frac{a_1a_2\sin \beta}{2} 
      \cos\bigl(\theta_1(x_N,\phi_j) -\theta_2(x_N,\phi_j)\bigr).
    \end{split}
\end{equation}
Thus, the map~\eqref{eq:change_var} is invertible everywhere except,
perhaps, on a surface $\theta_1= f(\theta_2;\beta)$ where the Jacobian
is zero.  Appropriate choices of the parameter $\beta$ in different
regions of $\Theta(N)$ allow to define the differential form $\omega$
everywhere in $\Theta(N)$. More explicitly, we have
\begin{equation}
  \label{eq:omega_form}
  \omega = \frac{2^{N\left(N-1\right)/2}}%
   {(2\pi)^N N!J(x,\phi_j)} 
     \prod_{1\le j < k \le N} \sin^2\left(\frac{\theta_k(x,\phi_j) -
        \theta_j(x,\phi_j)}{2}\right)\ed\phi_2 \wedge \dotsb \wedge
    \ed\phi_N. 
\end{equation}

Let $\psi$ and $\chi$ be two differential forms of degrees $p$ and $q$
respectively. The exterior derivative of $\psi \wedge \chi$ is a $(p+q
+ 1)$-form given by
\begin{equation}
  \label{eq:diff_prod}
  \ed\left(\psi \wedge \chi \right) = \ed\psi \wedge \chi +
  \left(-1\right)^p \psi \wedge \ed\chi .
\end{equation}
Now, $e^{i\xi X_N}/(i\xi)$ is a $0$-form on $\Theta(N)$.  Away from
the region where $J(x,\phi_j)=0$ the inverse of the
map~\eqref{eq:change_var} is differentiable with continuous
derivatives. Therefore, we have
\begin{equation}
  \label{eq:imt_part}
   e^{i \xi X_N} d\mu_\Theta = e^{i \xi X_N} \ed X_N \wedge \omega =
   \ed\left(\frac{e^{i\xi X_N}}{i\xi} \omega\right)  -
   \frac{e^{i\xi X_N}}{i\xi} \ed\omega.
\end{equation}
One can easily verify by direct calculation that $\omega$ is not
closed, \ie $\ed\omega \neq 0$.

If $\psi$ is a differential form of degree $p$ and $\Omega$ is a
manifold of dimension $p +1$, then Stokes' theorem states that
\begin{equation}
  \label{eq:stokes_theorem}
 \int_{\Omega} \ed\psi = \int_{\partial \Omega} \psi,
\end{equation}
where $\partial \Omega$ denotes the boundary of $\Omega$.  An
$N$-dimensional torus is a compact manifold without boundary,
therefore the right-hand side of~\eqref{eq:stokes_theorem} is 
zero.  As a consequence, integrating both sides of
equation~\eqref{eq:imt_part} we obtain
\begin{equation}
  \label{eq:asympt}
  \psi_N(\xi) = \int_{\Theta(N)}e^{i\xi X_N} \left(\ed X_N \wedge \omega\right)
 =  -\frac{1}{i\xi} \int_{\Theta(N)}e^{i\xi X_N}\ed \omega.
\end{equation}
It follows from the Riemann-Lebesgue lemma that
$\psi_N(\xi)=o\left(\xi^{-1}\right)$ as $\xi \to \infty$ and is
integrable. Thus, the inverse Fourier transform
\begin{equation}
   \label{eq:inv_four}
   f_N(x) = \frac{1}{2\pi}\int_{-\infty}^\infty e^{-i\xi
     x}\psi_N(\xi)d\xi  
\end{equation}
is well defined, bounded and uniformly continuous.
\end{proof}

\subsection{Smoothing}
\label{sse:sing_deg}

From the discussion in \S\ref{sssec:full_rank}, it follows that a
necessary (not sufficient) condition for
$\int_{S_N}^{N^\gamma}\abs{\psi_N\left(\xi\right)}d\xi$ to decay fast
enough is $\psi_N(\xi) =O\left(\xi^{-4}\right)$ as $\xi \to
\infty$. Indeed, if $A_N$ is of full rank and its spectrum is not
degenerate, equation~\eqref{eq:psiN_explicit} and the asymptotic
formula
\begin{equation}
  \label{eq:bessel_asympt}
  J_k(x) \sim \sqrt{\frac{2}{\pi x}}\cos\left(x - \tfrac{1}{4}k\pi -
    \tfrac{1}{4}\pi \right), \quad x \to \infty
\end{equation}
imply that $\psi_N\left(\xi\right) = O\bigl(\xi^{-N^2/2}\bigr)$.
Therefore, $f_N$ has continuous derivatives at least up to order
$N^2/2 - 2$.  In other words, $f_N$ becomes increasingly smooth as $N$
grows.

If a function has continuous derivatives of order $p$, then its
Fourier transform is $o\left(\xi^{-p}\right)$ as $\xi \to
\infty$. This suggests smoothing $F_N$ with an appropriate test
function.  More precisely, we define
\begin{equation}
  \label{eq:convolution}
  F^*_N(x):=\left[F_N * \chi_\epsilon\right](x) 
  = \int_{-\infty}^\infty F_N(t)\chi_\epsilon(x-t)dt,
\end{equation}
where $\chi_\epsilon \in \mathcal{C}_0^\infty\left(\mathbb{R}\right)$
and is normalized to one.  Our choice will be the test function
\begin{equation}
  \label{eq:chi_def}
  \chi_\epsilon(x) := \begin{cases} \frac{1}{g\epsilon}\exp\left(-\frac{1}{1 -
        \left(\frac{x}{\epsilon}\right)^2}\right) 
             & \text{if $x \in (-\epsilon,\epsilon)$,}\\
         0 & \text{if $x \in \mathbb{R} \setminus (-\epsilon,\epsilon)$,}
        \end{cases}
\end{equation}
where
\begin{equation}
  \label{eq:norm_chi}
  g:= \int_{-1}^1\exp\left(-\frac{1}{1 - x^2}\right)dx = 0.44399\dotso
\end{equation}

By differentiating $F^*_N$ and integrating by parts we obtain
\begin{equation}
  \label{eq:f_star}
  f^*_N(x) := F_N^{*\prime}(x)= \left[f_N * \chi_\epsilon\right](x) 
  = \int_{-\infty}^\infty f_N(t)\chi_\epsilon(x-t)dt.
\end{equation}

The convolution $f^*_N$ is positive and
\begin{equation}
  \label{eq:normalization}
  \int_{-\infty}^\infty f_N^*\left(x\right)dx = \left(\int_{-\infty}^\infty
  f_N(x)dx\right)\left( \int_{-\infty}^\infty
  \chi_\epsilon\left(x\right)dx\right) = 1.
\end{equation}
In addition, $f_N^* \in \mathcal{C}_0^\infty\left(\mathbb{R}\right)$
too and
\begin{equation}
  \label{eq:fourier_transform}
    \psi_N^*\left(\xi\right)  := \int_{-\infty}^\infty e^{i\xi
    x}f^*_N(x)dx = \psi_N(\xi)\hat{\chi}_\epsilon\left(\xi\right),
\end{equation}
where
\begin{equation}
  \label{eq:ft_chi}
  \hat{\chi}_\epsilon\left(\xi\right):=\int_{-\infty}^\infty  
   e^{i\xi x}\chi_\epsilon\left(x\right)dx.
\end{equation}

Let us introduce
\begin{alignat}{2}
  \label{eq:delta_star}
  \Delta(x) & := F_N(x) - \Phi(x),  \qquad &
  \Delta^*(x)  & := \left[\Delta * \chi_\epsilon\right](x).
\intertext{Then, we write}
\label{error_2}
e(N) &= \max_{x \in \mathbb{R}} \abs{\Delta(x)},
& e^*(N) & = \max_{x \in \mathbb{R}} \abs{\Delta^*(x)}.
\end{alignat}

Formula~\eqref{eq:feller} still holds if
we replace $\psi_N(\xi)$ and $\psi(\xi)$ with $\psi_N^*(\xi)$ and
$\psi^*(\xi):= \psi(\xi)\hat{\chi}_\epsilon\left(\xi\right)$ 
respectively. Indeed, let
$S_N=\delta N^{(1-b)/2}$ and $T_N = N^\gamma$ with $\gamma >2$. We have
\begin{equation}
  \label{eq:split_star}
  \begin{split}
    e^*(N) & \le \frac{2}{\pi} \int_{0}^{S_N}\left \lvert
    \frac{\psi^*_N(\xi) - \psi^*(\xi)}{\xi}\right \rvert d\xi + 
    \frac{2}{\pi S_N}\int_{S_N}^{N^\gamma}\abs{\psi_N^*\left(\xi\right)}d\xi \\
    & \quad  +   \frac{2}{\pi S_N}\int_{S_N}^{N^\gamma}
   \abs{\psi^*\left(\xi\right)}d\xi  + \frac{24m}{\pi N^\gamma},
  \end{split}
\end{equation}
where
\begin{equation}
  \label{eq:def_m}
  m :=\max_{\mathbb{R}} \abs{ \phi *
      \chi_\epsilon}.
\end{equation}
Now, 
\begin{equation}
  \label{eq:bound_abs_val}
  \abs{\psi^*_N - \psi^*} =\abs{\psi_N - \psi}\abs{\hat{\chi}_\epsilon}
  \le \abs{\psi_N - \psi}, 
\end{equation}
where we have used $\abs{\hat{\chi}_\epsilon} \le 1$, which holds for any
characteristic function. Therefore the Berry-Ess\'een
inequality~\eqref{eq:bound_1} applies to $\abs{\psi_N^* - \psi^*}$ too and
\begin{equation}
\frac{2}{\pi} \int_{0}^{S_N}\left \lvert
    \frac{\psi^*_N(\xi) - \psi^*(\xi)}{\xi}\right \rvert d\xi =
  O\left(N^{2-b}\right).
\end{equation}
Equation~\eqref{eq:def_err} gives
\begin{equation}
  \label{eq:bound_erfc}
  \int_{S_N}^{N^\gamma}\abs{\psi^*\left(\xi\right)}d\xi \le
   \int_{S_N}^{N^{\gamma}} e^{-\xi^2/2}d\xi \le
\frac{1}{2}\sqrt{\frac{\pi}{2}} \erfc\left(\delta N^{(1-b)/2}/\sqrt{2}
\right).
\end{equation}

In order to complete the proof of equation~\eqref{eq:error}, we need
to show that, for appropriate choices of the smoothing parameter
$\epsilon$, $e(N)\le Ce^*(N)$, and that the integral
\begin{equation}
  \label{eq:int_both}
  \int_{S_N}^{N^\gamma}\abs{\psi_N^*\left(\xi\right)}d\xi =
  \int_{S_N}^{N^\gamma}\abs{\psi_N\left(\xi\right)}\abs{\hat{\chi}_\epsilon
   \left(\xi\right)}d\xi 
\end{equation}
is sufficiently small. The appropriate choice of $\epsilon$ for which
these two statements are true is a delicate balance.  As $\epsilon$
decreases $e^*(N)$ will approach $e(N)$.  However, if the support of
$\chi_\epsilon$ is too small, its Fourier transform might spread for a
range of $\xi > S_N$ large enough to prevent the
integral~\eqref{eq:int_both} from decaying at a sufficiently fast
rate.

The leading order asymptotics of $\hat{\chi}_\epsilon(\xi)$ can be
computed using the method of steepest descent.  We report the
calculation in the appendix. We have
\begin{equation}
  \label{eq:leading_asymptotics}
  \begin{split}
  \hat{\chi}_\epsilon\left(\xi\right)& =
  \frac{2}{g(\epsilon\xi)^{3/4}}\sqrt{\frac{\pi}{\sqrt{2}}}
  \cos\left(\epsilon\xi - (\epsilon\xi)^{1/2} -
    \tfrac{3}{8}\pi\right)\\
    & \quad \times \exp\left(-(\epsilon\xi)^{1/2} -
    \tfrac{1}{4}\right)\Bigl
    (1 + O\left((\epsilon\xi)^{-1/2}\right)\Bigr), \quad
  \epsilon \xi \to \infty.
   \end{split}
\end{equation}
For this approximation to be meaningful $\xi > 1/\epsilon$.
Therefore, we cannot choose $\epsilon <C/S_N$, otherwise the bound on
the decay rate of the integral~\eqref{eq:int_both} would not be
adequate. 

It remains to establish if $\epsilon = O\left(S_N^{-1}\right)$ leads
to a good enough approximation to $e(N)$.  In order not to loose
information on the behaviour of $\Delta(x)$, the smoothing parameter
needs to be comparable with the rate of oscillation of $\Delta(x)$. In
other words, we need a bound on $\abs{\Delta'(x)}$.  Such a bound can
be obtained, once again, using the Berry-Ess\'een
inequality~\eqref{eq:bound_1}:
\begin{equation}
  \label{eq:bound_der}
  \begin{split}
  \abs{\Delta'(x)}& =\abs{f_N(x) - \phi(x)}\le
  \int_{-S_N}^{S_N}\abs{\psi_N(\xi) -\psi(\xi)}d\xi\\
  & \quad  + 2\int_{S_N}^\infty\abs{\psi_N(\xi)}d\xi + 2\int_{S_N}^\infty
  \abs{\psi(\xi)}d\xi. 
 \end{split}  
\end{equation}
By comparing this inequality with~\eqref{eq:split}, we see that the
first and third integral are $O\left(N^{-(2-b)}\right)$; the second
  integral might possibly be bigger by a factor $\pi S_N$.   It follows
  that
  \begin{equation}
    \label{eq:bound_deriv}
    \abs{\Delta ' (x)}\le C S_N e(N).
  \end{equation}
  This is sufficient for our purposes.  The following lemma completes
  the proof of theorem~\ref{th:main_theorem}.
  
\begin{lemma}
  \label{le:diff_Deltas}
  Suppose that $\abs{\Delta'(x)} \le e(N)/\eta_N$. Then, for $\epsilon
  = \eta_N$ there exists a positive constant $C =O(1)$ such that $e(N)
  \le C e^*(N)$.
\end{lemma}
\begin{proof}
  Since $\Delta(t)$ is continuous and $-1\le \Delta(t) \le 1$, there
  exists a $t_0 \in \mathbb{R}$ such that $e(N)=\Delta(t_0)$.  Then, we have
  \begin{equation}
    \label{eq:Del_bound_1}
    \Delta(t_0 + y) \ge e(N)\left(1 - \frac{y}{\eta_N}\right)
   , \quad \text{for
        $y >0$.} 
  \end{equation}

Now set
  \begin{equation}
    \label{eq:parameter}
    x  = t_0 + r\eta_N ,  \quad t = r\eta_N   - y,
  \end{equation}
  where $0 < r < 1$ is a parameter whose exact value is to be determined.
  Equation~\eqref{eq:Del_bound_1} becomes
    \begin{equation}
      \label{eq:Del_bound_2}
      \Delta\left(x - t \right) \ge e(N)\left(1 - r  +
        \frac{t}{\eta_N}\right), \quad 
      \text{for $\abs{t} \le r\eta_N$.} 
    \end{equation}
    Substituting this bound into the definition
  \begin{equation}
    \label{eq:del_star}
    \Delta^*(x) = \int_{-\infty}^{\infty}\Delta(x -t)\chi_{\eta_N}(t)dt
  \end{equation}
we arrive at
\begin{equation}
  \label{eq:inter}
e^*\left(N\right)\ge   \Delta^*(x) \ge e(N)\left(1 - r \right)  
   \int_{\abs{t} \le r \eta_N }\chi_{\eta_N}(t)dt  - e(N)\int_{\abs{t}
     \ge r \eta_N }\chi_{\eta_N}(t) dt,
\end{equation}
where we have used the inequality $\Delta(t) \ge -e(N) \ge -1$ and the
fact that the linear term in~\eqref{eq:Del_bound_2} does not
contribute because $\chi_{\eta_N}(t)$ is even.

Now, it turns out that the two integrals in~\eqref{eq:inter} are
independent of $\eta_N$.  Indeed,
\begin{equation}
  \label{eq:int_eq }
  \int_{\abs{t} \ge r \eta_N}
  \chi_{\eta_N}(t) dt = \frac{2}{g}\int_{r}^1
  \exp\left(-\frac{1}{1-t^2}\right)dt.
\end{equation}
Define
\begin{equation}
  \label{eq:hr}
  h(r) := r  + \frac{2}{g}\left(2 - r\right)\int_{r}^1
  \exp\left(-\frac{1}{1-t^2}\right)dt.
\end{equation}
Equation~\eqref{eq:inter} can now be rewritten
\begin{equation}
  \label{eq:new_hr}
  e^*(N) \ge e(N)\bigl(1-h(r)\bigr).
\end{equation}
Since $h(r)$ has a minimum
near $2/3$ and $h(2/3)=0.77646..\approx 4/5$, the statement of
the lemma follows with $C=5$.
\end{proof}

We achieve the rate of convergence in equation~\eqref{eq:error} if we
set 
\begin{equation}
    \label{def_eta} 
   \eta_N =  \frac{1}{C S_N} = \frac{1}{C\delta N^{\left(1-b\right)/2}}.
   \end{equation}

\section{Proof of theorem~\ref{th:total_variation}}
\label{se:tot_proof}

We need to prove equation~\eqref{eq:total_variation_2}. 

Let us choose the parameter $\eta_N$ in lemma~\ref{le:diff_Deltas} to
be 
\begin{equation}
  \label{eq:eta_eps}
  \eta_N = N^{-\zeta},
\end{equation}
where $0 < \zeta < \frac{1}{2}(1-b)$. Equation~\eqref{eq:bound_der} gives
\begin{equation}
  \label{eq:bound_der_3}
  \abs{f_N(x) - \phi(x)}\le C_\zeta N^{-2 + b + \zeta}
\end{equation}
for some constant $C_\zeta$.

% By definition we have
% \begin{equation}
%   \label{eq:pdf_char}
%   \abs{f_N(x)-\phi(x)}=\abs{\frac{1}{2\pi}\int_{-\infty}^\infty
%    \left(\psi_N(\xi)-
%       \psi(\xi)\right)e^{i\xi x}d\xi}.
% \end{equation}
% The integral on the right-hand side can be estimated in the same way
% as the integral in equation~\eqref{eq:feller}.  We obtain
% \begin{equation}
%   \label{eq:int_pass}
%   \begin{split}
%     \abs{f_N(x)-\phi(x)}& \le \frac{1}{2\pi}\int_{-\infty}^\infty
%     \abs{\psi_N(\xi)-\psi(\xi)}d\xi \\
%     & \le \frac{C}{\pi N^2}\int_0^{S_N}\xi^4e^{-\xi^2/2}d\xi +
%     \frac{1}{\pi}\int_{S_N}^\infty \abs{\psi_N(\xi)}d\xi +
%     \frac{1}{\pi}\int_{S_N}^\infty \abs{\psi(\xi)}d\xi  \\
%     & \le C_3 N^{-2}
%   \end{split}
% \end{equation}
% for some constant $C_3$.

Now take $\zeta^\prime >0$ and write
\begin{equation}
  \label{eq:fin_step}
  \begin{split}
    \int_{-\infty}^\infty \abs{f_N(x)-\phi(x)}dx & \le
    2\int_{0}^{N^{\zeta^\prime}}\abs{f_N(x)-\phi(x)}dx \\
    & \quad +
    2\int_{N^{\zeta^\prime}}^\infty f_N(x)dx + 2\int_{N^{\zeta^\prime}}^\infty
    \phi(x)dx \\
    & \le C_{\zeta, \zeta^\prime} N^{-2+b +\zeta +\zeta^\prime} 
    + 2\int_{N^{\zeta^\prime}}^\infty
    f_N(x)dx + 2\int_{N^{\zeta^\prime}}^\infty \phi(x)dx.
  \end{split}
\end{equation}
We have
\begin{equation}
  \label{eq:f_N}
  \begin{split}
  2\int_{N^{\zeta^\prime}}^\infty f_N(x)dx & =
  \int_{-\infty}^{-N^{\zeta^\prime}} f_N(x)dx + \int_{N^{\zeta^\prime}}^\infty
  f_N(x)dx \\
&     = F_N\left(-N^{\zeta^\prime}\right)  +
   1 - F_N\left(N^{\zeta^\prime}\right).
\end{split}
\end{equation}
Similarly
\begin{equation}
  \label{eq:phi_N}
  \begin{split}
  2\int_{N^{\zeta^\prime}}^\infty \phi(x)dx & =
  \int_{-\infty}^{-N^{\zeta^\prime}} \phi(x)dx + \int_{N^{\zeta^\prime}}^\infty
  \phi(x)dx \\
     & = \Phi\left(-N^{\zeta^\prime}\right)  +
   1 - \Phi\left(N^{\zeta^\prime}\right).
\end{split}
\end{equation}
Therefore, rearranging the terms and using the identity
$1-\Phi(x)=\Phi(-x)$ we obtain
\begin{equation}
  \label{eq:sum_int}
   \begin{split}
  2\int_{N^{\zeta^\prime}}^\infty
    f_N(x)dx + 2\int_{N^{\zeta^\prime}}^\infty \phi(x)dx & =
    \biggl(F_N\left(-N^{\zeta^\prime}\right) -
      \Phi\left(-N^{\zeta^\prime}\right)\biggr) \\ 
      & \quad -
     \biggl(F_N\left(N^{\zeta^\prime}\right) -
      \Phi\left(N^{\zeta^\prime}\right)\biggr)  
       + 4\biggl(1 - \Phi\left(N^{\zeta^\prime}\right)\biggr).
   \end{split}
\end{equation}
Finally, equation~\eqref{eq:total_variation_2} follows
from~\eqref{eq:error} and~\eqref{eq:ineq_er_fun}  by setting
$\epsilon =\zeta + \zeta^\prime$.

% If the singular values of $A_N$ are non-degenerate and different from
% zero and $O(1)$, then equations~\eqref{eq:third_int}
% and~\eqref{eq:bound_der} we can replace the
% bound~\eqref{eq:bound_der_3} with
% \begin{equation}
%   \label{eq:bound_der_4}
%   \abs{f_N(x) - \phi(x)}\le CN^{-2}.
% \end{equation}
% This leads to~\eqref{eq:total_variation_3}.

\section*{Acknowledgements}
We would like to express our gratitude to Paul Bourgade, Oliver
Johnson and Roman Schubert and for helpful discussions.  While this
research was carried out, F. Mezzadri was partially supported by EPSRC
grant no.  EP/G019843/1 and by a Leverhulme Research Fellowship.

\appendix

\section*{Appendix. The leading order asymptotics of
  $\hat{\chi}_\epsilon(\xi)$}
\renewcommand{\theequation}{A.\arabic{equation}}

The purpose of this appendix is to compute an explicit formula for the
leading order asymptotics of
\begin{equation}
  \label{eq:def_chi_a}
  \hat{\chi}_\epsilon\left(\xi\right) := 
   \int_{-\infty}^\infty e^{i\xi x}\chi_\epsilon (x)dx 
    = \frac{1}{g\epsilon} 
    \int_{-\epsilon}^\epsilon \exp\left(i\xi x -
  \frac{1}{1-\left(\frac{x}{\epsilon}\right)^2}\right)dx
\end{equation}
in the limit $\xi \to \infty$. Since $\hat{\chi}_\epsilon(\xi)=
\hat{\chi}_1\left(\epsilon\xi\right)$, for the sake of simplicity we
set $\epsilon=1$ and write $\hat{\chi} = \hat{\chi}_1$.  More
explicitly, we study the integral
\begin{equation}
  \label{eq:steep_desc_int}
  \hat{\chi}\left(\xi\right) =\frac{2}{g}\rpart \int_{0}^1 \exp\left(i\xi x -
  \frac{1}{1-x^2}\right)dx.
\end{equation}

\begin{propositiona}
We have
\begin{equation}
  \label{eq:leading_asymptotics_e1}
  \begin{split}
  \hat{\chi}\left(\xi\right) & =
  \frac{2}{g\xi^{3/4}}\sqrt{\frac{\pi}{\sqrt{2}}}
   \cos\left(\xi - \xi^{1/2} -
    \tfrac{3}{8}\pi\right)\\
    & \quad \times \exp\left(-\xi^{1/2} -
    \tfrac{1}{4}\right)\Bigl(1 + O\left(\xi^{-1/2}\right)\Bigr), \quad
  \xi \to \infty.
\end{split}
\end{equation}
\end{propositiona}
\begin{proof}
  \eqref{eq:steep_desc_int} can be estimated using the method
  of steepest descents.  The integrand is not analytic at one, so we
  look at
\begin{equation}
  \label{eq:int_delta}
  \frac{2}{g}\rpart \int_{0}^{1-\delta} \exp\left(i\xi x -
  \frac{1}{1-x^2}\right)dx,
\end{equation}
where $\delta > 0$ is small. The difference
between~\eqref{eq:steep_desc_int} and~\eqref{eq:int_delta} is bounded by
$2\delta e^{-1/(2\delta)}/g$.

Since we are interested only in the real part of~\eqref{eq:int_delta},
for large $\xi$ the origin will not contribute to leading order; the
main contribution should come from a small neighbourhood near one.

Consider the argument of the exponential:
\begin{equation}
  \label{eq:argument}
    f(x) := i\xi x - \frac{1}{1-x^2}.
\end{equation}
Its saddle points are the solutions of the equation
\begin{equation}
  \label{eq:sol_pol}
  i\xi\left(1- x^2\right)^2  - 2x = 0.
\end{equation}
For large $\xi$ the roots of this polynomial can be computed
perturbatively in the parameter $1/\xi$.  In other words, we look for
a solution near one with an asymptotic expansion of the form
\begin{equation}
  \label{eq:power_exp}
  \bar{x}(\xi) =  1 + \frac{x_1}{\xi^\alpha} + \frac{x_2}{\xi^{2\alpha}}
  + O\left(\xi^{-3\alpha}\right), \quad \xi \to \infty,
\end{equation}
where $\alpha$ is a rational power. By substituting this expression
into~\eqref{eq:sol_pol}, one finds that the two sides of the equation
can be balanced only if $\alpha=1/2$ and that the first two
coefficients are 
\begin{equation}
  \label{eq:next_to_leading}
  x_1 = \pm \frac{e^{\frac{3}{4}\pi i}}{\sqrt{2}} \quad \text{and}
  \quad x_2 =0.
\end{equation}

\begin{figure}[h]
\centering
\includegraphics[width=5.5in]{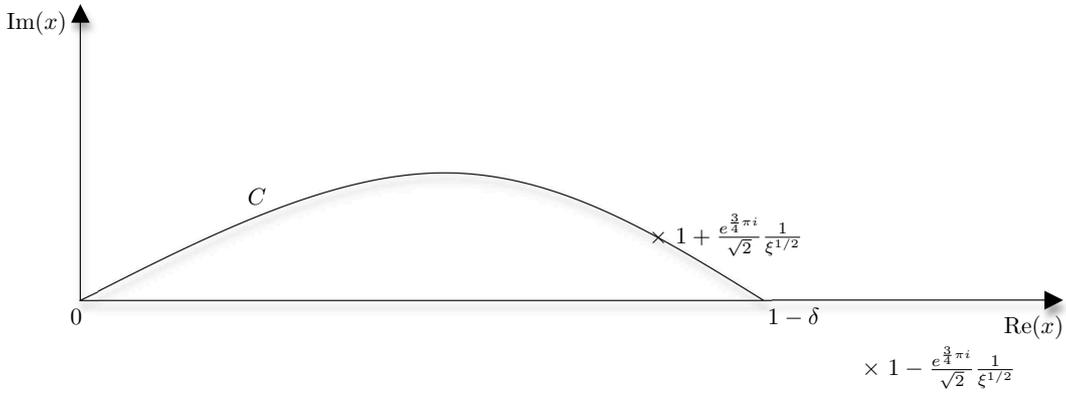}
\caption{The deformation of the interval $[0,1-\delta]$
and the saddle points of $f(x)$.} 
\label{fi:cont_def}
\end{figure}
Provided $\delta$ is sufficiently small, the interval of integration
of~\eqref{eq:int_delta} can be deformed into a contour asymptotically
equivalent to the steepest descent path passing through
\begin{equation}
  \label{eq:axxept_sol}
  \bar{x}(\xi) = 1 + \frac{e^{\frac{3}{4}\pi i}}{\sqrt{2}}\frac{1}{\xi^{1/2}} +
  O\left(\xi^{-3/2}\right), \quad \xi \to \infty.
\end{equation}
(See figure~\ref{fi:cont_def}.)
Such a deformation is not possible for the critical point
\begin{equation}
  \label{eq:axxept_sol_2}
  1 - \frac{e^{\frac{3}{4}\pi i}}{\sqrt{2}}\frac{1}{\xi^{1/2}} +
  O\left(\xi^{-3/2}\right), \quad \xi \to \infty.
\end{equation}
Trivial algebra gives
\begin{subequations}
  \label{eq:f_fpp}
\begin{align}
  f\bigl(\bar{x}(\xi)\bigr) & =i\left(\xi - \xi^{1/2}\right) 
  - \xi^{1/2} - \frac14 +
  O\left(\xi^{-1/2}\right), \quad \xi \to \infty, \\ 
  f^{\prime \prime}\bigl(\bar{x}(\xi)\bigr) & 
   = 2\sqrt{2}e^{-i\frac{\pi}{4}}\xi^{3/2} + 
  O\left(\xi^{1/2}\right), \quad \xi \to \infty. 
\end{align}
\end{subequations}
The tangent to the steepest descent path at $\bar x$ has equation
\begin{equation}
  \label{eq:tangent}
  x(t) = \bar{x}(\xi) + t e^{-\frac{3}{8}\pi i}.
\end{equation}
Therefore, we have
\begin{equation}
  \label{eq:lead_as3}
  \begin{split}
   \int_{0}^{1-\delta} \exp\left(i\xi x -
  \frac{1}{1-x^2}\right)dx & =
  \int_{C} \exp\left(i\xi x -
  \frac{1}{1-x^2}\right)dx \\
     &\sim e^{-\frac{3}{8}\pi i} 
         \sqrt{\frac{2\pi}{\abs{f^{\prime
        \prime}\bigl(\bar{x}(\xi)\bigr)}}} \exp\biggl(
     f\bigl(\bar{x}(\xi)\bigr)\biggr), \quad \xi \to \infty.
  \end{split}
\end{equation}
Finally, by inserting equations~\eqref{eq:f_fpp}
into~\eqref{eq:lead_as3} we arrive at~\eqref{eq:leading_asymptotics_e1},
provided $\delta < \xi^{-\beta}$ and $\beta > 3/4$.

\end{proof}

\vspace{.25cm}

\noindent\rule{16.2cm}{.5pt}

\vspace{.25cm}

{\small \noindent {\sl Department of Mathematics \\
                       University of Bristol\\
                       Bristol BS8 1TW, UK  \\
                       Email: {\tt j.p.keating@bristol.ac.uk}\\
                       Email: {\tt f.mezzadri@bristol.ac.uk}\\
                       Email: {\tt b.singphu@bristol.ac.uk}

                       \vspace{.25cm}

                       \noindent 16 November 2010}

\end{document}